\documentclass[11pt]{article}

\usepackage{amssymb,amsmath,amsthm,graphicx,framed,cleveref}
\usepackage{subfigure}
\usepackage{color}
\usepackage{fullpage}
\usepackage{url}
\usepackage[linesnumbered,ruled,vlined]{algorithm2e}
\usepackage{enumerate}
\usepackage{cite}
\newtheorem{definition}{Definition}
\newtheorem{lemma}{Lemma}

\newtheorem{corollary}{Corollary}
\newtheorem{theorem}{Theorem}
\newtheorem{observation}{Observation}

\newcommand{\memo}[1]{}

\title{{Complexity and algorithms for} finding\\ a perfect phylogeny from mixed tumor samples\thanks{This work has been submitted to the IEEE for possible publication. Copyright may be transferred without notice, after which this version may no longer be accessible.}}

\author{Ademir Hujdurovi\' c$^{a,b}$\thanks{E-mail: \texttt{ademir.hujdurovic@upr.si}}
\and
Ur\v sa~Ka\v car$^{c}$\thanks{E-mail: \texttt{ursa@celtra.com}}
\and
Martin Milani\v c$^{a,b}$\thanks{E-mail: \texttt{martin.milanic@upr.si}}
\and
Bernard Ries$^{d}$\thanks{E-mail: \texttt{bernard.ries@unifr.ch}}
\and
Alexandru I.~Tomescu$^{e}$\thanks{E-mail: \texttt{tomescu@cs.helsinki.fi}}
}

\begin{document}
\maketitle

\begin{center}
$^a$ University of Primorska, UP IAM, Muzejski trg 2, SI 6000 Koper, Slovenia\\

$^b$ University of Primorska, UP FAMNIT, Glagolja\v ska 8, SI 6000 Koper, Slovenia \\

$^c$ Celtra Inc., Ljubljana, Slovenia\\

$^d$ University of Fribourg, Department of Informatics Decision Support and Operations Research Group, Fribourg, Switzerland\\

$^e$ Helsinki Institute for Information Technology HIIT, Department of Computer Science, University of Helsinki, Finland\\
\end{center}

\bigskip

\begin{abstract}
Recently, Hajirasouliha and Raphael (WABI 2014) proposed a model for deconvoluting mixed tumor samples measured from a collection of high-throughput sequencing reads. This is related to understanding tumor evolution and critical cancer mutations. In short, their formulation asks to split each row of a binary matrix so that the resulting matrix corresponds to a perfect phylogeny and has the minimum number of rows among all matrices with this property. In this paper we disprove several claims about this problem, including an NP-hardness proof of it. However, we show that the problem is indeed NP-hard, by providing a different proof. We also prove NP-completeness of a variant of this problem proposed in the same paper.
On the positive side, we propose an efficient (though not necessarily optimal) heuristic algorithm based on coloring co-comparability graphs, and a polynomial time algorithm for solving the problem optimally on matrix instances in which no column is contained in both columns of a pair of conflicting columns. Implementations of these algorithms are freely available at \hbox{\url{https://github.com/alexandrutomescu/MixedPerfectPhylogeny}}.
\end{abstract}


\section{Introduction}\label{sec:introduction}

Tumor progression is assumed to follow a phylogenetic evolution in which each tumor cell passes its somatic mutations to its daughter cells as it divides, with new mutations being accumulated over time. It is important to discover what tumor types are present in the sample, at what evolutionary stage the tumor is in, or what are the ``founder'' mutations of the tumor, mutations that trigger an uncontrollable growth of the tumor. These can lead to better understanding of cancer \cite{Nik-Zainal:aa,campbell08}, better diagnosis, and more targeted therapies \cite{citeulike:12460879}.

DNA sequencing is one method for discovering the somatic mutations present in each tumor sample. The most accurate possible observation would come from sampling and sequencing every single cell. However, because of single-cell sequencing limitations, and the sheer number of tumor cells, one usually samples populations of cells. Even though the samples are taken spatially and morphologically apart, they can still contain millions of different cancer cells. Moreover, this mixing is not consistent across different collections of samples. Therefore, studying only these mixed samples poses a serious challenge to understanding tumors, their evolution, or their founding mutations.

Solutions for overcoming this limitation can come from a computational approach, as one could deconvolute each sample by exploiting some properties of the tumor progression. One common assumption is that all mutations in the parent cells are passed to the descendants. Another one, called the ``infinite sites assumption'', postulates that once a mutation occurs at a particular site, it does not occur again at that site. These two assumptions give rise to the so-called \emph{perfect phylogeny} evolutionary model. Hajirasouliha and Raphael proposed in~\cite{DBLP:conf/wabi/HajirasoulihaR14} a model for deconvoluting each sample into a set of tumor types so that the \emph{multi-set} of all resulting tumor types forms a perfect phylogeny, and is minimum with this property. Even though this model has some limitations, for example it assumes no errors, and only single nucleotide variant mutations, it is a fundamental problem whose understanding can lead to more practical extensions.

Other major approaches for deconvoluting tumor heterogeneity include methods based on somatic point mutations, such as PyClone~\cite{citeulike:13162167}, SciClone~\cite{citeulike:13321876}, {PhyloSub}~\cite{Jiao:2014aa}, CITUP~\cite{Malikic01052015}, LICHeE~\cite{Popic2015}, and methods based on somatic copy number alterations, such as THetA~\cite{THetA}, TITAN~\cite{Ha01112014} and MixClone~\cite{DBLP:journals/bmcgenomics/LiX15}.

Let us review two methods from the first category mentioned above. CITUP~\cite{Malikic01052015} exhaustively enumerates through all possible phylogenetic trees (up to maximum number of vertices) and tries to decompose each sample into several nodes of the phylogeny. The fit between each sample and the phylogenetic tree is one minimizing a Bayesian information criterion on the frequencies of each mutation. This is computed either exactly, with quadratic integer programming, or with a heuristic iterative method. The tree achieving an optimal fit is output, together with the decompositions of each sample as nodes (i.e., sets of mutations) of this tree.

Method LICHeE~\cite{Popic2015} also tries to fit the observed mutation frequencies to an optimal phylogenetic tree, but with an optimized search for such a tree. Mutations are first assigned to clusters based on their frequencies (a mutation can belong to more clusters). These clusters form the nodes of a directed acyclic graph (DAG). Directed edges are added to these graphs from a node to all its possible descendants, based on inclusions among their mutation sets and on compatibility among their observed frequencies. Spanning trees of this DAG are enumerated, and the ones best compatible with the mutation frequencies are output.

As opposed to the problem proposed in \cite{DBLP:conf/wabi/HajirasoulihaR14} and considered in this paper, these two methods heavily rely on the mutation frequencies. Frequencies appear at the core of their problem formulations, and without them they would probably output an arbitrary phylogenetic tree and a decomposition of samples into nodes of this tree compatible with the observed data. The two problems in this paper have the same output, but only assume data on absence or presence of mutations in each sample. The objective function of our first problem requires that the sum, over each sample, of the number of leaves of the phylogeny that the sample is decomposed to, is minimum. The second problem formulation only requires that output phylogeny has the minimum number of leaves.

In this paper we show that several proofs from~\cite{DBLP:conf/wabi/HajirasoulihaR14} related to the optimization problem
proposed therein are incorrect, including an NP-hardness proof of it.\footnote{Hajirasouliha and Raphael mentioned during their WABI 2014 talk that their claim about every graph being a row-conflict graph (Theorem 4 in~\cite{DBLP:conf/wabi/HajirasoulihaR14}) contained a flaw and proposed a correction stating that for every binary matrix $M$ with an all-zeros row and an all-ones row, the complement of $G_{M,r}$ (for any row $r$ of $M$) is transitively orientable (cf.~Section~\ref{sec:problem-formulation} for the definition of $G_{M,r}$ and Theorem~\ref{thm:GMr} below). In particular, the fact that
Theorem 4 in~\cite{DBLP:conf/wabi/HajirasoulihaR14} is incorrect implies that the NP-hardness proof from~\cite{DBLP:conf/wabi/HajirasoulihaR14}
is incorrect.} However, the NP-hardness claim turns out to be correct: in this paper we provide a different NP-hardness proof. We also adapt this proof to a variant of the problem also proposed in~\cite{DBLP:conf/wabi/HajirasoulihaR14} but whose complexity was left open. This problem asks to minimize the \emph{set} (instead of multi-set) of all tumor types of the perfect phylogeny.
The two problems, formally defined in Section~\ref{sec:problem-formulation}, are referred to as the
{\sc Minimum Conflict-Free Row Split} and the {\sc Minimum Distinct Conflict-Free Row Split} problem, respectively.

Moreover, we obtain a polynomial time algorithm for a collection of instances of the {\sc Minimum Conflict-Free Row Split} problem, which can be biologically characterized as follows. Say that two mutations $i$ and $j$ are \emph{exclusive} if $i$ is present in a sample in which $j$ is absent, and $j$ is present in a sample in which $i$ is absent. Observe that exclusive mutations cannot both be present in the same vertex of a perfect phylogeny. Thus, we say that a sample is a \emph{mixture} at exclusive mutations $i$ and $j$ if both $i$ and $j$ are present in that sample. The instances for which we can solve the problem in polynomial time are such that for any two exclusive mutations $i$ and $j$, no mutation is present only in the samples mixed at $i$ and $j$.

\begin{sloppypar}
We also propose an efficient (though not necessarily optimal) heuristic algorithm for the {\sc Minimum Conflict-Free Row Split} problem,
based on coloring co-comparability graphs, and provide implementations of both algorithms, freely available at \url{https://github.com/alexandrutomescu/MixedPerfectPhylogeny}.
\end{sloppypar}

\medskip
\noindent {\bf Paper outline.} In Section~\ref{sec:problem-formulation} we give all formal definitions and review the approach of~\cite{DBLP:conf/wabi/HajirasoulihaR14}. In Section~\ref{sec:characterization} we give a complete characterization of the so-called {\it row-conflict graphs},
the class of graphs considered in~\cite{DBLP:conf/wabi/HajirasoulihaR14}. The complexity results are presented in Section~\ref{sec:complexity-results}, and the above-mentioned polynomial time algorithm is given in Section~\ref{sec:no-containments-in-conflicting}.
In Section~\ref{sec:heuristic} we discuss the heuristic algorithm for general instances, and in Sections~\ref{sec:experiments} and~\ref{sec:discussion} we present experimental results on the binary matrices from clear cell renal cell carcinomas (ccRCC) from~\cite{Gerlinger:2014aa}.
We conclude the paper with a discussion in Section~\ref{sec:conclusion}.

Some of the results in this paper appeared in the proceedings of WABI 2015~\cite{WABI2015}.
In addition to the material presented in~\cite{WABI2015}, this paper contains all the missing proofs
(complete proofs of Theorems 2, 3, and 4, a more detailed proof of Lemma 2), a time complexity analysis of
the algorithm presented in Section 5, a discussion following the proof of Theorem 5 on the necessity of the assumptions
for the algorithm given in Section 5, and three additional sections (Sections 6, 7 and 8) describing a
polynomial time heuristic for general instances and experimental results.

\section{Problem formulation\label{sec:problem-formulation}}

As mentioned in the introduction, we assume that we have a set of sequencing reads from each tumor sample, and that based on these reads we have discovered the sample variants with respect to a reference (e.g., by using a somatic mutation caller such as VarScan~2~\cite{Koboldt02022012}). This gives rise to an $m \times n$ matrix $M$ whose $m$ rows are the different samples, and whose $n$ columns are the genome loci where a mutation was observed with respect to the reference. The entries of $M$ are either 0 or 1, with 0 indicating the absence of a mutation, and 1 indicating the presence of the mutation. We assume that the matrix has no row whose all entries are 0.

Under ideal conditions, e.g., each mutation was called without errors, and the samples do not contain reads from several leaves of the perfect phylogeny, $M$ corresponds to a perfect phylogeny matrix. Such matrices are characterizable by a simple property, called \emph{conflict-freeness}.

\begin{definition}
Two columns $i$ and $j$ of a binary matrix $M$ are said to be {\em in conflict}
if there exist three rows $r, r', r''$ of $M$ such that $M_{r,i} = M_{r,j}= 1$,
$M_{r',i}=M_{r'',j} = 0$, and $M_{r',j} = M_{r'',i}= 1$.
A binary matrix $M$ is said to be {\em conflict-free} if no two columns of $M$ are in conflict.
\end{definition}
\noindent It is well known that the rows of $M$ are leaves of a perfect phylogenetic tree if and only if $M$ is conflict-free (see~\cite{Estabrook:1975aa,gusfieldbook}). Moreover, if this is the case, then the corresponding phylogenetic tree can be retrieved from $M$
 in time linear in the size of $M$~\cite{Gusfield91}.

However, in practice,  each tumor sample is a mixture of reads from several tumor types, and thus possibly $M$ is not conflict-free. If we are not allowed to edit the entries of $M$ as done e.g. by methods such as \cite{Rens:aa}, \cite{DBLP:journals/jcb/SalariSHKNWSB13},
Hajirasouliha and Raphael proposed in~\cite{DBLP:conf/wabi/HajirasoulihaR14} to turn $M$ into a conflict-free matrix $M'$ by splitting each row $r$ of $M$ into some rows $r_1,\dots,r_k$ such that $r$ is the bitwise OR of $r_1,\dots,r_k$; that is, for every column $c$, $M_{r,c} = 1$ if and only if $M_{r_i,c} = 1$ for at least one $r_i$. The rows $r_1,\dots,r_k$ can be seen as the deconvolution of the mixed sample $r$ into samples from single vertices of a perfect phylogeny. One can then build the perfect phylogeny corresponding to $M'$ and carry further downstream analysis. Let us make this row split operation precise.

\begin{sloppypar}
\begin{definition}
Given a binary matrix $M\in \{0,1\}^{m\times n}$ with rows labeled $r_1,r_2,\ldots, r_m$, we say that a binary matrix $M'\in \{0,1\}^{m'\times n}$ is a {\em row split} of $M$ if there exists a partition of the set of rows of $M'$ into $m$ sets $R'_1,R'_2,\ldots,R'_m$ such that for all $i \in \{1,2,\ldots,m\}$, $r_i$ is  the bitwise OR of the binary vectors given by the rows of $R'_i$. The set $R_i'$ of rows of $M'$ is said to be a set of {\em split rows} of row $r_i$.
\end{definition}
\end{sloppypar}

Observe that a simple strategy for obtaining a conflict-free row split of $M$ is to split every row $r$ into as many rows as there are 1s in $r$, with a single 1 per row. While this might be an informative solution for some instances (cf.~also Corollary~\ref{cor:no-containments} on p.~\pageref{cor:no-containments}),
Hajirasouliha and Raphael proposed in~\cite{DBLP:conf/wabi/HajirasoulihaR14} as criterion for obtaining a meaningful conflict-free
row split $M'$ the requirement that the number of rows of $M'$ is minimum among all conflict-free row splits of $M$.

In this paper we consider the following problem, which we call {\sc Minimum Conflict-Free Row Split} problem.
For a binary matrix $M$, we denote by $\overline{\gamma}(M)$ the minimum number of rows
in a conflict-free row split $M'$ of $M$.
This notation is in line with notation $\gamma(M)$ used in~\cite{DBLP:conf/wabi/HajirasoulihaR14} to denote the minimum number of
{\it additional} rows in a conflict-free row split $M'$ of $M$, that is, $\gamma(M) = \overline{\gamma}(M)-m$, where $m$ is the number of rows of $M$.

\vbox{\begin{framed}
\noindent {\sc Minimum Conflict-Free Row Split}:\\
\noindent {\bf Input:} A binary matrix $M$, an integer $k$.\\
\noindent {\bf Question:} Is it true that $\overline{\gamma}(M)\leq k$?
\end{framed}}

\smallskip
The optimization version of the above problem (in which only a given subset of rows needs to be split) was called the Minimum-Split-Row problem in~\cite{DBLP:conf/wabi/HajirasoulihaR14}, however, all results from~\cite{DBLP:conf/wabi/HajirasoulihaR14} deal with the variant of the problem in which all rows need to be split (some perhaps trivially by setting $R_i' = \{r_i\}$), which is equivalent to the {\sc Minimum Conflict-Free Row Split} problem.

Given a binary matrix $M$ and a row $r$ of $M$, the {\it conflict graph of $(M,r)$} is the graph
$G_{M,r}$ defined as follows: with each entry $1$ in $r$, we associate a vertex in $G_{M,r}$, and
two vertices in $G_{M,r}$ are connected by an edge if and only if the corresponding
columns in $M$ are in conflict. Denoting by $\chi(G)$ the chromatic number of a graph $G$,
Hajirasouliha and Raphael proved in~\cite{DBLP:conf/wabi/HajirasoulihaR14}
the following lower bound on the value of $\overline{\gamma}(M)$:

\begin{lemma}\cite{DBLP:conf/wabi/HajirasoulihaR14}\label{lem:lower-bound}
Let $M$ be a binary matrix $M$ with a conflict-free row split $M'$.
Then, for every row $r_i$ of $M$ with a set $R_i'$ of split rows of $M'$, we have
$|R_i'|\ge \chi(G_{M,r_i})$.
\end{lemma}

\begin{corollary}\label{cor:WABI}
For every binary matrix $M$, we have $\overline{\gamma}(M)\ge \sum_{r}\chi(G_{M,r})$.
\end{corollary}

Hajirasouliha and Raphael also claimed in~\cite{DBLP:conf/wabi/HajirasoulihaR14} the following hardness result.

\begin{theorem}\cite{DBLP:conf/wabi/HajirasoulihaR14}
\label{thm:WABI}
The {\sc Minimum Conflict-Free Row Split} problem is NP-hard.
\end{theorem}

To recall their approach for proving Theorem~\ref{thm:WABI}, we need one more definition.
We denote the fact that two graphs $G$ and $H$ are isomorphic by $G\cong H$.

\begin{definition}\label{def:RC}
A graph $G$ is a {\em row-conflict graph} if there exists a binary matrix $M$ and a row $r$ of $M$ such that $G\cong G_{M,r}$.
\end{definition}

The proof of Theorem~\ref{thm:WABI} was based on a reduction from the chromatic number problem in graphs and relied on three ingredients: the lower bound given by Corollary~\ref{cor:WABI},
Theorem 4 from~\cite{DBLP:conf/wabi/HajirasoulihaR14} stating that every graph is a row-conflict graph,
and an algorithm based on graph coloring, also proposed in~\cite{DBLP:conf/wabi/HajirasoulihaR14},
for optimally solving the {\sc Minimum Conflict-Free Row Split} problem by constructing a conflict-free row split of $M$
with exactly $\sum_{r}\chi(G_{M,r})$ rows. In particular, their results would imply that the lower bound on $\overline{\gamma}(M)$ given by Corollary~\ref{cor:WABI}
is always attained with equality.

Contrary to what was claimed in~\cite{DBLP:conf/wabi/HajirasoulihaR14}, we show that there exist graphs that are not row-conflict graphs. In fact, we give a complete characterization of row-conflict graphs, showing that a graph is a row-conflict graph if and only if its complement is transitively orientable (see \Cref{thm:GMr}). Using a reduction from $3$-edge-colorability of cubic graphs, we show that it is NP-complete to test whether a given binary matrix $M$ has a conflict-free row split $M'$ with number of rows achieving the lower bound given by
Corollary~\ref{cor:WABI} (see \Cref{thm:hard}). This implies that there exist infinitely many matrices for which this bound is not achieved.

A corollary of our characterization of row-conflict graphs is that the chromatic number is polynomially computable for this class of graphs.
This fact with the assumption that P\,$\neq$\,NP, as well as the existence of matrices $M$ with  $\overline{\gamma}(M)> \sum_{r}\chi(G_{M,r})$, each individually imply that the claimed NP-hardness proof of the {\sc Minimum Conflict-Free Row Split} problem given in~\cite{DBLP:conf/wabi/HajirasoulihaR14} is flawed.
Nevertheless, our NP-completeness proof (see Theorem~\ref{thm:hard}) implies that Theorem~\ref{thm:WABI} is correct.

On the positive side, we give a polynomial time algorithm for the {\sc Minimum Conflict-Free Row Split} problem on input matrices $M$ in which
no column is contained in both columns of a pair of conflicting columns (see Theorem~\ref{thm:poly}).

We also consider a variant of the problem, also proposed in~\cite{DBLP:conf/wabi/HajirasoulihaR14}, in which we are only interested in
minimizing the number of {\it distinct} rows in a conflict-free row split of $M$. This problem is similar to the Minimum Perfect Phylogeny Haplotyping problem~\cite{DBLP:journals/jcb/BafnaGHY04}, in which we need to explain a set of genotypes with a minimum number of haplotypes admitting a perfect phylogeny. For a binary matrix $M$, we denote by $\overline{\eta}(M)$ the minimum number of {\it distinct} rows in a conflict-free row split $M'$ of $M$. We establish NP-completeness of the following problem (see Theorem~\ref{thm:hard2}), which was left open in~\cite{DBLP:conf/wabi/HajirasoulihaR14}.

\vbox{\begin{framed}
\noindent {\sc Minimum Distinct Conflict-Free Row Split}:\\
\noindent {\bf Input:} A binary matrix $M$, an integer $k$.\\
\noindent {\bf Question:} Is it true that $\overline{\eta}(M)\leq k$?
\end{framed}}

\section{A characterization of row-conflict graphs\label{sec:characterization}}

\begin{definition}\label{def:H_M}
Given a binary matrix $M$ and two columns $i$ and $j$ of $M$, column $i$ is said to be {\em contained in} column $j$ if
 $M_{k,i}\le M_{k,j}$ holds for every $k$. The  {\em undirected containment graph} $H_M$ is the undirected graph
whose vertices correspond to the columns of $M$ and in which two vertices $i$ and $j$, $i\neq j$, are adjacent if and only if the column corresponding to vertex $i$ is contained in the column corresponding to vertex $j$ or vice versa.
\end{definition}

Recall that an {\it orientation} of an undirected graph $G=(V,E)$ is a directed graph $D=(V,A)$ such that for every edge $uv\in E$, either
$(u,v)\in A$ or $(v,u)\in A$, but not both. An orientation is said to be {\it transitive} if the presence of the directed edges $(u,v)$ and $(v,w)$ implies the presence of the directed edge $(u,w)$.  A graph is said to be {\it transitively orientable} if it has a transitive orientation.
The {\it complement} of a graph $G$ is a graph $\overline{G}$ with the same vertex set as $G$ in which two distinct vertices are adjacent if and only if they are non-adjacent in $G$. Transitively orientable graphs appeared in the literature under the name of \textit{comparability graphs} (and their complements under the name of \textit{co-comparability graphs}). Transitively orientable graphs and their complements form a subclass of the well known class of perfect graphs~\cite{MR2063679}. Therefore, odd cycles of length at least $5$ and their complements are examples of graphs that are not transitively orientable.

\begin{observation}\label{obs:transitive}
For every binary matrix $M$, the graph $H_M$ is transitively orientable.
\end{observation}

\begin{proof}
We say that column $i$ is {\em properly contained in} column $j$ if
$i$ is contained in $j$ and $M_{k,i} < M_{k,j}$ for some $k$.
Fix an ordering $\{c_1,\ldots, c_n\}$ of the columns of $M$.
Let us define a binary relation $\sqsubset$ on the set of columns on $M$ by setting,
for every two columns $c_i$ and $c_j$ of $M$, $c_i \sqsubset c_j$ if and only if
either $c_i$ is properly contained in $c_j$, or
$i<j$ and each of $c_i$ and $c_j$ is contained in the other one (that is, as binary vectors they are the same).
Observe that for a pair of columns $c_i$ and $c_j$ with $c_ic_j\in E(H_M)$
we have either $c_i\sqsubset c_j$ or $c_j\sqsubset c_i$ but not both.
The binary relation $\sqsubset$ defines an orientation of $H_M$, by orienting each edge $c_ic_j$ as going from $c_i$ to $c_j$ if and only if $c_i \sqsubset c_j$. This orientation can be easily verified to be transitive.
\end{proof}

In the next theorem, we characterize row-conflict graphs (cf.~Definition~\ref{def:RC}).

\begin{theorem}\label{thm:GMr}
A graph $G$ is a row-conflict graph if and only if $\overline{G}$ is
transitively orientable.
\end{theorem}

\begin{proof}
($\Rightarrow$) Let $M$ be an arbitrary binary matrix, $r$ an arbitrary row of $M$, and let $G=G_{M,r}$. Let $N$ be the submatrix of $M$ consiting of the columns of $M$ that have $1$ in row $r$. It is now easy to see that ${G}_{M,r}\cong {G}_{N,r}$. Moreover, any two columns of $N$ are either in conflict or their corresponding vertices are adjacent in $H_N$. Therefore, $H_N\cong \overline{{G}_{N,r}}$. Since $H_N$ is transitively orientable (by Observation~\ref{obs:transitive}), it follows that $\overline{G}$ is transitively orientable as well.

($\Leftarrow$) We follow the strategy of the proof of Theorem 4 in~\cite{DBLP:conf/wabi/HajirasoulihaR14} (which works for complements
of transitively orientable graphs). For the sake of completeness, we include here a short proof of this implication.
Let $G$ be a graph such that $H=\overline{G}$ is transitively orientable, with a transitive orientation $\overrightarrow{H}$.
It can be easily seen that $\overrightarrow{H}$ is acyclic, thus we may assume that vertices of $G$ are topologically ordered as $V(G)=\{v_1,\ldots,v_n\}$, that is,
for every directed edge $(v_i,v_j)$ in $\overrightarrow{H}$, we have $i<j$.
Let $E(G)=\{e_1,e_2,\ldots, e_m\}$. We construct a matrix $M$ with $n$ columns and $2m+1$ rows, such that ${G}_{M,1}\cong G$. The first row of $M$ is defined to have all entries equal to 1. For every edge $e_k = v_iv_j$, $i<j$, of $G$, the $2k$-th row of $M$ has entry $0$ in the column corresponding to vertex $v_i$, and entry $1$ in the column corresponding to $v_j$. Additionally, the $(2k+1)$-st row of $M$ has entry $1$ in the column corresponding to vertex $v_i$, and entry $0$ in the column corresponding to $v_j$. Since the first row has all entries equal to 1, after filling in these entries of $M$, the two columns corresponding to $v_i$ and $v_j$, respectively, are in conflict.

We need to fill in the remaining entries of $M$ so that we do not introduce any new conflicts. For every $i$, we fill in the remaining entries so that whenever $(v_i,v_j)$ is a directed edge in $\overrightarrow{H}$, the column corresponding to the vertex $v_i$ is contained in the column corresponding to the vertex $v_j$. This can be achieved by examining the columns one by one, following the topological order $(v_1,\ldots, v_n)$ of $\overrightarrow{H}$, and filling each unfilled entry with a $0$, unless this would violate the above containment principle.

At the end of this procedure, there are no conflicts between columns corresponding to vertices $v_i$ and $v_j$, whenever $(v_i,v_j)$ is a directed edge in $\overrightarrow{H}$. Therefore, ${G}_{M,1}\cong G$.
%
\end{proof}

Theorem~\ref{thm:GMr} implies that odd cycles of length at least $5$ and their complements are not row-conflict graphs.
The reader not familiar with transitively orientable graphs might find it useful to verify that the cycle of length $5$ cannot be transitively oriented.

\section{Complexity results\label{sec:complexity-results}}

\begin{theorem}\label{thm:hard}
The following two problems are NP-complete:
\begin{itemize}
  \item The {\sc Minimum Conflict-Free Row Split} problem.
  \item Given a binary matrix $M$, is it true that $\overline{\gamma}(M) = \sum_{r}\chi(G_{M,r})$?
\end{itemize}
\end{theorem}

\begin{proof}
The {\sc Minimum Conflict-Free Row Split} problem is in NP, since testing if a given binary matrix $M'$ with at most $k$ rows, equipped with a partition of its rows into $m$ sets, satisfies the condition in the definition of a row split, as well as the conflict-freeness, can be done in polynomial time. To argue that the second problem is in NP, we proceed similarly as above, performing an additional test checking that the number of rows of $M'$ equals $\sum_{r}\chi(G_{M,r})$. (In this case, we will have $\overline{\gamma}(M) \le \sum_{r}\chi(G_{M,r})$ and equality will follow from~Corollary~\ref{cor:WABI}.) The value of $\sum_{r}\chi(G_{M,r})$ can be computed in polynomial time,
since each graph $G_{M,r}$ is the complement of a transitively orientable graph (by Theorem~\ref{thm:GMr}), and the chromatic number of complements of transitively orientable graphs can be computed in polynomial time (see, e.g.,~\cite{MR2063679}).

We prove hardness of both problems at once, making a reduction from the following NP-complete problem~\cite{MR635430}:
Given a simple cubic graph $G = (V,E)$, is $G$ $3$-edge-colorable?
(A graph is {\it cubic}, or {\it $3$-regular}, if every vertex is incident with precisely three edges.
A {\it matching} in a graph is a set of pairwise disjoint edges.
A graph is {\it $3$-edge-colorable} if its edge set can be partitioned into $3$
matchings.)

Given a simple cubic graph $G = (V,E)$, we construct an instance $(M,k)$ of the
{\sc Minimum Conflict-Free Row Split} problem as follows:
\begin{itemize}
  \item $M$ is a $(|V|+3)\times (|E|+3)$ binary matrix, with rows indexed by $V\cup\{r_1,r_2,r_3\}$,
  columns indexed by $E\cup\{c_1,c_2,c_3\}$, and entries defined as follows
(see Fig.~\ref{fig:example} for an example):
  \begin{itemize}
    \item For every row indexed by a vertex $v\in V$ and every column indexed by an edge $e$, we have $$M_{v,e} = \left\{
                                                                                                                  \begin{array}{ll}
                                                                                                                    1, & \hbox{if $v$ is an endpoint of $e$;} \\
                                                                                                                    0, & \hbox{otherwise.}
                                                                                                                  \end{array}
                                                                                                                \right.$$
    \item For every row indexed by a vertex $v\in V$ and every column indexed by some $c\in \{c_1,c_2,c_3\}$, we have $M_{v,c} = 1$.
    \item For every row indexed by some $r\in \{r_1,r_2,r_3\}$ and every column indexed by an edge $e\in E$, we have $M_{r,e} = 0$.
    \item For every row indexed by some $r_i\in \{r_1,r_2,r_3\}$ and every column indexed by some $c_j\in \{c_1,c_2,c_3\}$,
    we have $$M_{r_i,c_j} = \left\{
                             \begin{array}{ll}
                               1, & \hbox{if $i = j$} \\
                               0, & \hbox{otherwise.}
                             \end{array}
                           \right.$$
      \end{itemize}
  \item $k = 3|V|+3$.
\end{itemize}

\begin{figure}[!ht]
  \begin{center}
\includegraphics[width=0.66\linewidth]{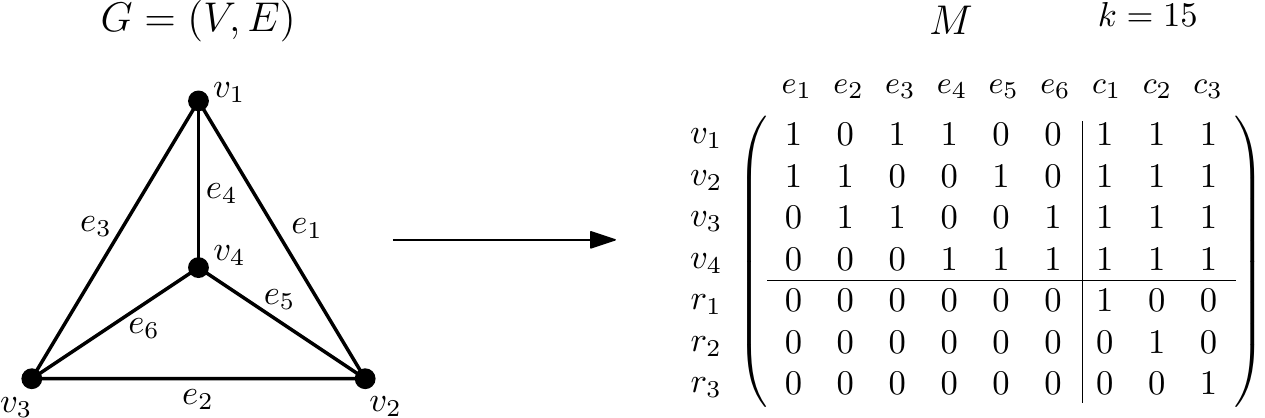}
  \end{center}
\caption{An example construction of $(M,k)$ from $G$.} \label{fig:example}
\end{figure}

Note that for each row indexed by a vertex $v\in V$, the graph $G_{M,v}$ is isomorphic to the disjoint union of
two complete graphs with three vertices each, hence $\chi(G_{M,v})= 3$.
For each row indexed by some $r\in \{r_1,r_2,r_3\}$, the graph $G_{M,r}$ consists in a single vertex, thus
$\chi(G_{M,r})= 1$. It follows that $k = \sum_{r}\chi(G_{M,r})$ and therefore $M$ is a yes instance to the
second problem (``Given a binary matrix $M$, is $\overline{\gamma}(M) = \sum_{r}\chi(G_{M,r})$?'') if and only if
$(M,k)$ is a yes instance for the {\sc Minimum Conflict-Free Row Split} problem.
Hardness of both problems will therefore follow from the following claim, which we prove next:
{\it $G$ is $3$-edge-colorable if and only if $\overline{\gamma}(M)\leq k$. }

Suppose first that $G$ is $3$-edge-colorable, and let $E = E_1\cup E_2\cup E_3$ be a partition of $E$ into $3$  matchings.
We obtain a row split $M'$ of $M$ by replacing each row of $M$ indexed by a vertex $v\in V$ with three rows
and keeping each row of $M$ indexed by some $r\in \{r_1,r_2,r_3\}$ unchanged.
Clearly, this will result in a matrix with $k$ rows.
For every $v\in V$, we replace the row of $M$ indexed by $v$ as follows.
Vertex $v$ is incident with precisely three edges in $G$, say $e_1,e_2,e_3$.
Since $E_1,E_2,E_3$ are matchings partitioning $E$, we may assume, without loss of generality, that $e_i\in E_i$ for all $i\in\{1,2,3\}$.
The three rows replacing in $M'$ the row of $M$ indexed by $v$
are indexed by $v^1$, $v^2$, $v^3$ and defined as follows:
for every $i\in \{1,2,3\}$ and every column $c\in E\cup\{c_1,c_2,c_3\}$, we have
$$M'_{v^i,c} = \left\{
                \begin{array}{ll}
                  1, & \hbox{if $c = e_i$ or $c = c_i$;} \\
                  0, & \hbox{otherwise.}
                \end{array}
              \right.$$
By construction, $M'$ is a row split of $M$ with $k$ rows. We claim that $M'$ is conflict-free.
No pair of columns indexed by two edges in $E$ agree on value $1$ in any row, hence they cannot be in conflict.
The same holds for any pair of columns indexed by two elements of $\{c_1,c_2,c_3\}$.
Consider now two columns, one indexed by an edge $e\in E$ and one indexed by $c_i\in \{c_1,c_2,c_3\}$.
Without loss of generality, we may assume that $e\in E_1$.
There are only two rows in which the column indexed by $e$ has value $1$, namely the rows indexed by copies of the endpoints of $e$,
say $u^1$ and $v^1$ (with $u,v\in V$).
The values of $M'$ in column $c_i$ at rows $u^1$ and $v^1$ are both $1$  (if $i =1$), otherwise they are both $0$.
Consequently, the two columns cannot be in conflict.
 Since $M'$ is a conflict-free row split of $M$ with $k$ rows, this establishes $\overline{\gamma}(M) \le k$.

For the converse direction, let $M'$  be a conflict-free row split of $M$ with at most $k$ rows.
Let $V' = V\cup \{r_1,r_2,r_3\}$ and consider a partition $\{R'_i\mid i\in V'\}$ of the set of rows of $M'$ into $|V|+3$ sets
indexed by elements of $V'$ such that
for all $i\in V'$, the row of $M$ indexed by $i$ is the bitwise OR of the rows of $R'_i$.
Since $k$ is a lower bound on $\overline{\gamma}(M)$, matrix $M'$ has exactly $k$ rows.
This fact and Corollary~\ref{cor:WABI} imply that each row in $M$ indexed by a vertex
$v\in V$ has $|R_v'| = 3$ and each row indexed by some $r\in \{r_1,r_2,r_3\}$ has $|R_r'| = 1$.

We must have that for all $i\in V'$, the row of $M$ indexed by $i$ is the bitwise {\it sum} of the rows of $R'_i$,
that is, for every column $c\in E\cup\{c_1,c_2,c_3\}$, we have
$M_{i,c} = \sum_{r\in R_i'}M'_{r,c}$.
Indeed, if for some $i\in V'$ and some column $c\in E\cup\{c_1,c_2,c_3\}$, we have that
$\sum_{r\in R_i'}M'_{r,c}>1$, then $i$ is a vertex of $G$. Furthermore, since
$|R_i'| = 3$, there are either two edges of $G$, say $e$ and $f$, incident with $i$
such that for some $r\in R_i'$, we have $M'_{r,e} = M'_{r,f} = 1$, or there are
two distinct elements $e,f\in \{c_1,c_2,c_3\}$ with the same property.
In the former case, considering the rows replacing the rows of $M$ indexed by the endpoints of $e$ and $f$ other than $i$,
 respectively, we find two distinct rows $r'$ and $r''$ of $M$ such that
$M'_{r',e} = M'_{r'',f} = 1$ and $M'_{r'',e} = M'_{r',f} = 0$, which
contradicts the fact that $M'$ is conflict-free.
In the latter case, the argument is similar.

By permuting the rows of $M'$ if necessary, we may assume that each set of the form
$R_v'$ is ordered as $R_v' = \{v^1,v^2,v^3\}$ so that $$M'_{v^i,c_j} = \left\{
                                                                        \begin{array}{ll}
                                                                          1, & \hbox{if $i = j$;} \\
                                                                          0, & \hbox{otherwise.}
                                                                        \end{array}
                                                                      \right.$$
We claim that for every edge $e = uv\in E$ and every $i\in \{1,2,3\}$, we have that $M'_{u^i,e} = M'_{v^i,e}$.
If this was not the case, then we would have $M'_{u^i,e} = M'_{v^j,e} = 1$ for a distinct pair $i,j\in \{1,2,3\}$.
But then the columns of $M'$ indexed by $e$ and $c_i$ would both agree
in value $1$ in row indexed by $u^i$ and disagree (in opposite directions) in rows indexed by
$v^j$ and $v^i$. Thus, they would be in conflict, contrary to the fact that $M'$ is conflict-free.

Since for every edge $e = uv\in E$ and every $i\in \{1,2,3\}$, we have that $M'_{u^i,e} = M'_{v^i,e}$,
we can partition the edges of $E$ into three pairwise disjoint sets $E_1,E_2,E_3$ by placing every edge
$e = uv\in E$ into $E_i$ if and only if $i\in \{1,2,3\}$ is the unique index such that
$M'_{u^i,e} = M'_{v^i,e} = 1$.
We claim that each $E_i$ is a matching in $G$. This will imply that $G$ is $3$-edge-colorable and complete the proof.
If some $E_i$ is not a matching, then there exist two distinct edges, say $e,f\in E_i$ with a common endpoint.
Let $e = xy$ and $f = xz$. The columns of $M'$ indexed by $e$ and $f$ agree in value $1$ at
row indexed by $x^i$, while they disagree (in opposite directions) in rows indexed by
$y^i$ and $z^i$. Thus, they are in conflict, contrary to the conflict-freeness of $M'$.
\end{proof}

Hajirasouliha and Raphael proposed in~\cite{DBLP:conf/wabi/HajirasoulihaR14} an algorithm based on graph coloring for optimally solving the {\sc Minimum Conflict-Free Row Split} problem by constructing a conflict-free row split of $M$ with exactly $\sum_{r}\chi(G_{M,r})$ rows.
Since there are infinitely many cubic graphs that are not $3$-edge-colorable (see, e.g.,~\cite{MR0382052}),
the proof of Theorem~\ref{thm:hard} implies that there exist infinitely many matrices $M$ such that $\overline{\gamma}(M)>\sum_r\chi(G_{M,r})$.
On such instances, the algorithm from~\cite{DBLP:conf/wabi/HajirasoulihaR14} will not produce a valid (that is, conflict-free) solution.

Since the smallest cubic $4$-edge-chromatic graph is the Petersen graph, the smallest matrix $M$ with $\overline{\gamma}(M)>\sum_{r}\chi({G}_{M,r})$ that can be obtained using the construction given in the proof of Theorem~\ref{thm:hard} is of order $13 \times 18$. A smaller matrix $M$ for which the bound from
Corollary~\ref{cor:WABI} is not tight can be obtained by
applying a similar construction starting from the complete graph of order $3$ (which is a $2$-regular $3$-edge-chromatic graph):
$$\small
M=\left(
  \begin{array}{ccc|cc}
      1 & 1 & 0 & 1 & 1\\
      1 & 0 & 1 & 1 & 1\\
      0 & 1 & 1 & 1 & 1\\
      \hline
      0 & 0 & 0 & 1 & 0\\
      0 & 0 & 0 & 0 & 1\\
  \end{array}
\right).
$$
We leave it as an exercise for the reader to verify that $\sum_{r}\chi({G}_{M,r})=8$ and
$\overline{\gamma}(M)\ge 9$ (in fact, $\overline{\gamma}(M)= 9$).
Let us also remark that in~\cite[Section 4.2.1]{Kacar} a binary matrix $M$ is given with $\overline{\gamma}(M)=\sum_{r}\chi({G}_{M,r})$,
on which the algorithm from~\cite{DBLP:conf/wabi/HajirasoulihaR14} fails to produce a conflict-free solution.

We conclude this section with another hardness result.

\begin{theorem}\label{thm:hard2}
The {\sc Minimum Distinct Conflict-Free Row Split} problem is NP-complete.
\end{theorem}

\begin{proof}
Membership in NP of the {\sc Minimum Distinct Conflict-Free Row Split} problem
can be argued similarly as for the {\sc Minimum Conflict-Free Row Split} problem.
It suffices to argue that there is a polynomially-sized conflict-free matrix $M'$ such that
$M'$ is a row split of $M$ with at most $k$ distinct rows.
We may assume that for a partition $R_1',\ldots, R_m'$ of rows of $M'$ into $m$ sets satisfying the condition in the definition of a row split,
the rows within each $R_i'$ are pairwise distinct. Recall (e.g.~from~\cite{gusfieldbook}) that a conflict-free matrix
with $d$ distinct rows and $n$ columns corresponds to a perfect phylogenetic rooted tree $T$ such that: $T$ has $d$ leaves (the rows of the matrix), all internal vertices of $T$ are branching, and all edges from a vertex to its children are injectively labeled with a column of $M$, with the exception of at most one edge which is unlabeled. Thus $T$ has at most $2n$ edges, and we infer that $d\le 2n$.
Therefore, the total number of rows of $M'$ does not exceed $2nm$, where $m$ and $n$ are the numbers of rows and columns of $M$, respectively.

The hardness proof is based on a slight modification of the reduction used in the proof of
Theorem~\ref{thm:hard}. (See Fig.~\ref{fig:example2} for an example.)
Given a cubic graph $G=(V,E)$, we map it to $(\overline{M},\overline{k})$ where
\begin{itemize}
\item $\overline{M}$ is the binary matrix obtained from the binary matrix $M$ described in the proof of Theorem~\ref{thm:hard}
by adding to it three columns $d_1,d_2,d_3$, which on the rows indexed by $V$ equal $0$, and on the rows indexed by $r_1,r_2,r_3$, each $d_i$ equals $c_i$, $i \in \{1,2,3\}$.
\item $\overline{k} = |E|+3$.
\end{itemize}

\begin{figure}[!ht]
  \begin{center}
\includegraphics[width=0.66\linewidth]{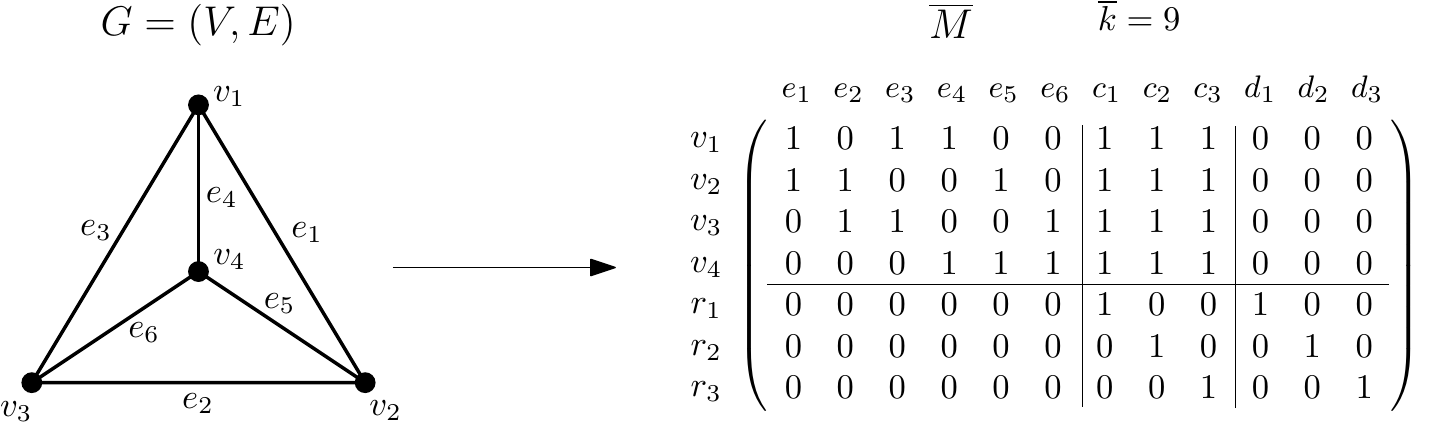}
  \end{center}
\caption{An example construction of $(\overline M,\overline k)$ from $G$.} \label{fig:example2}
\end{figure}

\begin{sloppypar}
We claim that $(\overline{M},\overline{k})$ is an instance of the {\sc Minimum Distinct Conflict-Free Row Split} problem
such that $G$ is $3$-edge-colorable if and only if \hbox{$\overline{\eta}(\overline{M})\le \overline{k}$}.
\end{sloppypar}

Suppose that $G$ is $3$-edge-colorable. Given a partition of $E$ into three matchings $E = E_1 \cup E_2 \cup E_3$, we construct the same matrix $M'$ as described in the proof of Theorem~\ref{thm:hard}, to which we append the three columns indexed by $d_1,d_2,d_3$ which are all $0$s on the rows indexed by vertices, and which are the same as in $\overline{M}$ on the rows $r_1,r_2,r_3$. By the same argument given in the proof of Theorem~\ref{thm:hard}, $M'$ is conflict-free. Each row $r_i$, $i\in \{1,2,3\}$, is distinct from all other rows of $M'$. Let $v^i$, $i \in \{1,2,3\}$, be a row corresponding to a vertex $v$ and suppose $M'_{v^i,e} = 1$, where $e=uv$ is one of the three edges incident to $v$, and $e \in E_i$. By construction, the only other row having a $1$ in column $e$ is $u^i$. Thus, row $v^i$ is different from all other rows, except $u^i$. In fact, we can see that row $v^i$ is identical to row $u^i$, since they have no other entry $1$ on the columns indexed by edges. Additionally, they both have $1$ in column $c_i$, since $e \in E_i$, and $0$ in the other five columns in $\{c_1,c_2,c_3,d_1,d_2,d_3\} \setminus \{c_i\}$. Hence, the number of distinct rows of $M'$ is at most
$3|V|/2 + 3 = |E| + 3 = \overline{k}$, since $G$ is cubic, and thus $\overline{\eta}(\overline M)\le \overline{k}$.

For the converse direction, suppose that $M'$ is a conflict-free row split of $\overline{M}$ with at most $\overline{k} = |E| + 3$ distinct rows.
Let $V' = V\cup \{r_1,r_2,r_3\}$ and consider a partition $\{R'_i\mid i\in V'\}$ of the set of rows of $M'$ into $|V|+3$ sets
indexed by elements of $V'$ such that for all $i\in V'$, the row of $\overline M$
indexed by $i$ is the bitwise OR of the rows of $R'_i$.
We will prove that (1) the number of pairwise distinct rows in $R_v'$ is $3$ for all $v \in V$, and that (2) the number of pairwise distinct rows in $R_{r}$ is $1$ for all $r \in \{r_1,r_2,r_3\}$.
Applying the same approach as in the proof of Theorem~\ref{thm:hard} will then imply that $G$ is $3$-edge-colorable.

As argued in the proof of Theorem~\ref{thm:hard}, no row in $R'_v$ has two $1$s on two columns indexed by two edges, say $e$ and $f$, because each of $e$ and $f$ has an endpoint which is not an endpoint of the other edge (and thus
a row with two $1$s on two columns indexed by two edges would imply a conflict in $M'$).
Moreover, no row in $R'_v$ has two $1$s on two columns indexed by $c_1,c_2,c_3$.

Let us associate with each row of $M'$ belonging to some $R_i'$ with $i\in V$ the edge column where it has a $1$ (if there is any). Since each edge column contains a $1$ and no row has two $1$s on the columns indexed by edges, the number of pairwise distinct rows of $M'$ indexed by a vertex is at least $|E|$. Since in each $R_{r_i}'$, $i \in \{1,2,3\}$, we must have at least one row distinct from all other rows of $M'$ (because of the $1$s in columns $d_1,d_2,d_3$), and $M'$ has at most $|E| + 3$ pairwise distinct rows, the number of distinct rows of $M'$ is exactly $|E| + 3$. This directly implies (2), more precisely, that each $R_{r}$ consists only of a row identical to the corresponding row of $\overline{M}$.

In order to prove property (1), suppose now that there is a row of $\overline M$ indexed by a vertex $v$ such that $R_v'$ contains at least $4$ pairwise distinct rows. Observe first that there is no row in $R_v'$ having a $1$ only in one column among $\{c_1,c_2,c_3\}$ (and only 0s in the columns indexed by edges).
Indeed, besides being distinct from the row in each $R_{r}'$, $r \in \{r_1,r_2,r_3\}$, it would also be distinct from each of the set of at least $|E|$ rows of $M'$ having a $1$ on a column indexed by an edge. Thus this would contradict the fact that $M'$ has at most $|E| + 3$ pairwise distinct rows. This implies that there are two distinct rows $v'$ and $v''$ in $R_v'$ such that $v'$ and $v''$ both contain a $1$ on the same column indexed by an edge, say $e$, but on a column among $\{c_1,c_2,c_3\}$, say $c_i$, $v'$ contains $1$ and $v''$ contains $0$. This shows that there is a conflict in $M'$, since $M'_{r_i,e} = 0$ and $M'_{r_i,c_i} = 1$, a contradiction.
\end{proof}

\section{A polynomially solvable case}\label{sec:no-containments-in-conflicting}

In this section we consider the binary matrices in which
no column is contained in both columns of a pair of conflicting columns, and derive a polynomial time algorithm  for the {\sc Minimum Conflict-Free Row Split} problem on such matrices. The main idea behind the algorithm
is the fact that on such matrices the lower bound from
Corollary~\ref{cor:WABI}
is achieved,
and the bound can be expressed in terms of
parameters of a set of derived digraphs, the so-called directed containment graph (see Definition~\ref{def:dcg} below).

Let $M$ be a binary matrix such that no column of $M$ is contained in two or more conflicting columns.
If there are duplicated columns in $M$, then we form a new matrix where we take just one copy of the columns that are duplicated.
Since an optimal solution of the reduced instance can be mapped to an optimal solution of the original instance
(by duplicating the columns corresponding to the copies of the duplicated columns in $M$ kept by the reduction),
we may assume that there are no duplicated columns in $M$.

\begin{definition}\label{def:dcg}
Given a binary matrix $M$ with distinct columns $c_1,\dots,c_n$ and a row $r$ of $M$, the \em directed containment graph of $(M,r)$ is the graph
$\overrightarrow{H}_{M,r}$ whose vertex set is the set of columns of $M$ having a $1$ in row $r$, in which there is a directed edge from $c_i$ to $c_j$ if and only if $i\neq j$ and $c_i$ is contained in $c_j$.
\end{definition}

\noindent We will use the notation $c_i \sqsubset_r c_j$ as a shorthand for the fact that $(c_i,c_j)$ is a directed edge of $\overrightarrow{H}_{M,r}$. We say that $c_i$ is a {\em source} of $\overrightarrow{H}_{M,r}$ if $c_i \in V(\overrightarrow{H}_{M,r})$ and there is no $c_j$ with $c_j \sqsubset_r c_i$. Let $\sigma(M,r)$ denote the number of sources in $\overrightarrow{H}_{M,r}$.

\begin{lemma}\label{lem:delta-chi}
If there are no duplicated columns in $M$, then
$\sigma(M,r)\leq \chi(G_{M,r})$ holds for any row $r$ of $M$.
\end{lemma}
\begin{proof}
Two vertices in the complement of $G_{M,r}$ are adjacent if and only if the corresponding columns of $M$ are either disjoint or
one contains the other one. However, since the vertices of both $\overrightarrow{H}_{M,r}$ and $G_{M,r}$
correspond to columns in which $M$ has value $1$ in row $r$, no two such columns can be disjoint. Consequently,
the underlying undirected graph of $\overrightarrow{H}_{M,r}$ is
equal to the complement of $G_{M,r}$. The set of all sources of $\overrightarrow{H}_{M,r}$ forms an independent set in its underlying undirected graph. This set corresponds to a clique in $G_{M,r}$.
Therefore  $\sigma(M,r)\leq \omega(G_{M,r})\leq \chi(G_{M,r})$ (where $\omega(G_{M,r})$ denotes the maximum size of a clique in $G_{M,r}$).
\end{proof}

Our algorithm is the following one (see also Fig.~\ref{fig:algorithm} for an example).

\begin{framed}
\noindent
{\bf Input}: An $m\times n$ binary matrix $M$ with columns $c_1,c_2,\ldots, c_n$, without duplicated columns, and such that no column of $M$ is contained in both columns of a pair of conflicting columns.

\noindent
{\bf Output}: A conflict-free row split $M'$ of $M$ with $\overline{\gamma}(M)$ rows.

\noindent
{\bf Algorithm:}
\begin{enumerate}
  \item Define a new matrix $M'$ with columns $c'_1,c'_2,\ldots, c'_n$.
  \item For each row $r$ of $M$, add the rows $r'_1,\ldots,r'_{\sigma(M,r)}$ to $M'$, defined as:

~~~~~Let $c_{r,1},\ldots,c_{r,\sigma(M,r)}$ be the sources of $\overrightarrow{H}_{M,r}$.

~~~~~$M'_{r'_i,c'_j} = \left\{
                                                                        \begin{array}{ll}
                                                                          1, & \hbox{if $c_{r,i} = c_j$ or $c_{r,i} \sqsubset_r c_j$;}  \\
                                                                          0, & \hbox{otherwise.}
                                                                        \end{array}
                                                                      \right.$
  \item Return $M'$.
\end{enumerate}
\end{framed}

\begin{figure}[!ht]
\centering
\includegraphics[width=0.7\linewidth]{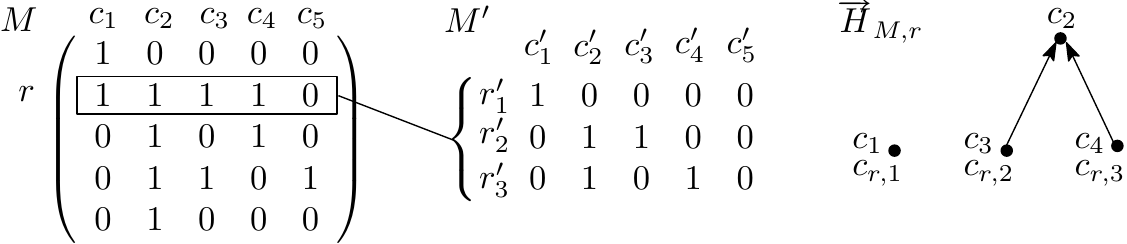}
\caption{An example of a matrix $M$ in which
 no column is contained in both columns of a pair of conflicting columns ($c_1,c_2$ and $c_3,c_4$ are conflicting). The rows $r'_1,r'_2,r'_3$ constructed by the algorithm corresponding to row $r$ of $M$ are shown in the center. On the right, the directed containment graph of $(M,r)$.\label{fig:algorithm}}
\end{figure}

\begin{theorem}
\label{thm:poly}
For any $m\times n$ binary matrix $M$ without duplicated columns such that no column of $M$ is contained in
both columns of a pair of conflicting columns, it holds that $\overline{\gamma}(M) = \sum_{r}\chi(G_{M,r}) = \sum_r\sigma(M,r)$.
{Moreover, a conflict-free row split $M'$ of $M$ with $\overline{\gamma}(M)$ rows can be constructed in time $O(mn^2)$.}
\end{theorem}

\begin{proof}
We claim that the matrix $M'$ produced by the above algorithm is a conflict-free row split of $M$ with number of rows equal to $\overline{\gamma}(M)$.

It is clear that $M'$ is a row split of $M$.
Let us prove that $M'$ is conflict-free. Suppose the contrary, that is, let $c'_i$ and $c'_j$ be two columns of $M'$ which are in conflict. Then, there exists a row $r'_k$ of $M'$ (obtained by splitting a row $r$ of $M$) which has $1$ in columns $c'_i$ and $c'_j$.

We will first show that $c_i$ is contained in $c_j$ or vice versa.
If $c_{r,k}=c_i$ (resp.~$c_{r,k}=c_j$) then $c_i \sqsubset_r c_j$ (resp.~$c_j \sqsubset_r c_i$)
and therefore column $c_i$ is contained in column $c_j$ (resp.~$c_j$ is contained in $c_i$).
Suppose now that $c_{r,k}\not \in \{c_i,c_j\}$.
Since $r'_k$ has 1 in columns $c'_i$ and $c'_j$ it follows that $c_{r,k} \sqsubset_r c_i$ and $c_{r,k} \sqsubset_r c_j$. This implies that column $c_{r,k}$ is contained in both column $c_i$ and column $c_j$. By the assumption on $M$, $c_i$ and $c_j$ cannot be in conflict, hence, one of them is contained in the other one.

Thus, we may assume without loss of generality that $c_i$ is contained in $c_j$.
Since $c'_i$ and $c'_j$ are in conflict it follows that there exists a row $w'_\ell$ of $M'$ which has $1$ in column $c'_i$ and $0$ in column $c'_j$.
This implies that the corresponding row $w$ of $M$ has $1$ in column $c_i$, and consequently also in $c_j$, since $c_i$ is contained in $c_j$.
Therefore, both $c_i$ and $c_j$ are vertices of $\overrightarrow{H}_{M,w}$.
If $c_i=c_{w,\ell}$, then $w'_\ell$ has value $1$ in column $c'_j$ (since $c_i$ is contained in $c_j$), which contradicts the choice of $w'_\ell$.
Thus, $c_i\ne c_{w,\ell}$ and $c_{w,\ell} \sqsubset_w c_i$.
However, since $c_i$ is contained in $c_j$ and $\overrightarrow{H}_{M,w}$ is transitive, it follows that $c_{w,\ell} \sqsubset_w c_j$. This implies that row $w'_\ell$ has value $1$ in column $c'_j$, which again contradicts the choice of $w'_\ell$. This finally shows that $M'$ is conflict-free.

Since the number of rows in $M'$ is $\sum_r \sigma(M,r)$ and $M'$ is conflict free, we have
$\overline{\gamma}(M)\le \sum_r \sigma(M,r)$.
By Corollary~\ref{cor:WABI} and Lemma~\ref{lem:delta-chi} we have
$\sum_r \sigma(M,r)\leq \sum_r\chi(G_{M,r}) \leq \overline{\gamma}(M)$.
This implies equality.

{It remains to justify the time complexity.
First, we compute, in time $O(mn^2)$, the transitive orientation
$\overrightarrow{H}_{M}$ of the undirected containment graph $H_M$
as specified by Observation~\ref{obs:transitive} (that is, $(c_i, c_j)$ is an arc of
$\overrightarrow{H}_{M}$ if and only if $c_i$ is properly contained in $c_j$).
Since for each row $r$ of $M$, the graph
$\overrightarrow{H}_{M,r}$ is an induced subdigraph of $\overrightarrow{H}_{M}$,
the $\sigma(M,r)$ sources of $\overrightarrow{H}_{M,r}$ can be computed from $\overrightarrow{H}_{M}$ in
the straightforward way in time $O(n^2)$.
The corresponding $\sigma(M,r)$ new rows of $M'$ can be computed in
time $O(\sigma(M,r)n)$, which results in total time complexity of
$O(mn^2)+O(\sum_{r}\sigma(M,r)n)= O(mn^2)$, as claimed. }
\end{proof}

Note that the correctness of the algorithm crucially relies on the assumption that
no column of the input matrix is contained in both columns of a pair of conflicting columns.
For example, the algorithm fails to resolve the conflict in the $3\times 3$ input matrix $M = \left(
                                                                                   \begin{array}{ccc}
                                                                                     1 & 1 & 1 \\
                                                                                     0 & 1 & 0 \\
                                                                                     0 & 0 & 1 \\
                                                                                   \end{array}
                                                                                 \right)$, which violates the assumption (column $c_1$ is contained in
                                                                                 both columns $c_2$ and $c_3$, which are in conflict).
Given the above matrix $M$, the output matrix $M'$ computed by the algorithm is in fact equal to $M$.

It is also worth mentioning that if the input matrix satisfies the stronger property that no column is contained in another one, Theorem~\ref{thm:poly} implies that the na\"{i}ve solution obtained by splitting each row $r$ into as many $1$s as it contains always produces an optimal solution. This is true since all vertices of $\overrightarrow{H}_{M,r}$ are sources. We thus obtain:

\begin{corollary}\label{cor:no-containments}
For any binary $m \times n$ matrix $M$ such that no column of $M$ is contained in another one, it holds that $\overline{\gamma}(M) = m'$, where $m'$ equals the number of 1s in $M$. Moreover, a conflict-free row split $M'$ of $M$ of size $m' \times n$ can be constructed in time $O(m'n)$.
\end{corollary}

{
\section{A heuristic algorithm based on coloring co-comparability graphs}\label{sec:heuristic}

As pointed out in Section~\ref{sec:complexity-results}, the graph theoretic algorithm from~\cite[Section 4]{DBLP:conf/wabi/HajirasoulihaR14}
fails to always produce a conflict-free row split of the input matrix. In this section, we propose a polynomial time heuristic algorithm for the
{\sc Minimum Conflict-Free Row Split} problem, that is, an algorithm that always produces a conflict-free row split of the input matrix.
This algorithm is also based on graph colorings.

Before presenting the algorithm, we describe the intuition behind it.
The lower bound on $\overline{\gamma}(M)$ given by Corollary~\ref{cor:WABI}
follows from the fact that in every conflict-free row split $M'$ of the input matrix $M$,
the rows replacing row $r$ in the split can be used to
produce a valid vertex coloring of $G_{M,r}$, the conflict graph of $(M,r)$.
The difficulty in reversing this argument in order to obtain a row split of $M$ having a number of rows close
to the lower bound $\sum_r\chi(G_{M,r})$ is due to the fact that we cannot independently combine the
splits of rows $r$ of $M$ according to optimal colorings of their conflict graphs, as new conflicts may arise.

We can guarantee that the corresponding row splits will be pairwise compatible (in the sense that no new conflicts
will be generated) as follows. We color $G_M$, the conflict graph of the whole input matrix (which we will define in a moment), and
split each row $r$ according to the coloring of its conflict graph $G_{M,r}$ given by the restriction of the coloring of $G_M$ to the vertex set of
$G_{M,r}$. The graph $G_M$, {\it the conflict graph of $M$}, is defined as follows: with each column of $M$,
we associate a vertex in $G_{M}$. Two vertices in $G_{M}$ are connected by an edge if and only if the corresponding
columns in $M$ are in conflict. Note that each conflict graph $G_{M,r}$ of an individual row is an induced subgraph of $G_M$, hence the restriction of
any proper coloring $c$ of $G_M$ to $V(G_{M,r})$ is a proper coloring of $G_{M,r}$.

The above approach will result in a row split having number of rows given by the value of $\sum_{r}|c(V(G_{M,r}))|$,
where $|c(V(G_{M,r}))|$ denotes the number of colors used by $c$ on $V(G_{M,r})$. As a first heuristic attempt to minimize this quantity,
Ka\v car proposed in~\cite[Section 4.2.2]{Kacar} to choose a coloring $c$ of $G_M$ with $\chi(G_M)$ colors. However, this is computationally intractable. While row-conflict graphs are characterized (in Theorem~\ref{thm:GMr}) as exactly the co-comparability graphs (that is, as complements of transitively orientable graphs)---for which the coloring problem is polynomially solvable~\cite{MR2063679}---, conflict graphs of binary matrices do not enjoy such nice features. Indeed, if $G$ is any graph of minimum degree at least $2$, then $G\cong G_M$, where $M\in \{0,1\}^{E(G)\times V(G)}$ is the edge-vertex incidence matrix of $G$ (defined by $M_{e,v} = 1$ if and only if vertex $v$ is an endpoint of edge $e$).

This can be amended as follows. We can ``restore'' the structure of co-comparability graphs by observing that $G_M$ is a spanning subgraph of $\overline{H_M}$, the complement of the undirected containment graph $H_M$ (cf.~Definition~\ref{def:H_M}), and working with $\overline{H_M}$ instead. Recall that $H_M$ is the undirected graph whose vertices correspond to the columns of $M$ and in which two vertices $i$ and $j$, $i\neq j$, are adjacent if and only if one the
corresponding columns is contained in the other one.  To show that $G_M$ is a spanning subgraph of $\overline{H_M}$, note first that we may assume that $V(G_M) = V(\overline{H_M})$ (as both vertex sets are in bijective correspondence with the set of columns of $M$). Moreover, if two vertices $i$ and $j$ of $G_M$ are adjacent, then the corresponding columns are in conflict, which implies that neither of them is contained in the other one; consequently, they are adjacent in $\overline{H_M}$.

Since $G_M$ is a spanning subgraph of $\overline{H_M}$, any proper coloring of $\overline{H_M}$ is also a proper coloring of $G_M$.
Moreover, even though the graph $\overline{H_M}$ might have more edges than $G_M$, these additional edges (if any)
will not be contained in any of the graphs $G_{M,r}$. Indeed, for every row $r$, its conflict graph $G_{M,r}$ coincides both with the subgraph of $G_M$ induced by $U := V(G_{M,r})$ as well as with the subgraph of $\overline{H_M}$ induced by $U$. This is because for any two vertices $i$ and $j$ in $U$ that are adjacent in $\overline{H_M}$, the corresponding columns cannot be disjoint, therefore, since $i$ and $j$ are not adjacent in $H_M$, the corresponding columns must be in conflict.

In view of the above observations, we propose choosing an optimal coloring $c$ of
the co-comparability graph $\overline{H_M}$ as a heuristic approach to minimizing the
value of $\sum_{r}|c(V(G_{M,r}))|$ for a coloring $c$ of $G_M$. A
row split of $M$ is then defined according to the coloring $c$.

This leads to the following algorithm (see Fig.~\ref{fig:algorithm-heuristic} for an example):

\medskip
\begin{framed}
\noindent
{\bf Input}: An $m\times n$ binary matrix $M$ with columns labeled with $1,\ldots, n$.

\noindent
{\bf Output}: A conflict-free row split $M'$ of $M$.

\noindent
{\bf Algorithm:}
\begin{enumerate}
  \item Define a new matrix $M'$ with $n$ columns labeled with $1,\ldots, n$.
  \item Compute $\overline{H_M}$, the complement of the undirected containment graph $H_M$.
  \item Compute an optimal coloring $c$ of $\overline{H_M}$.
  \item For each row $r$ of $M$:

  ~~~~~Let $c(V(G_{M,r})) = \{s_1^r,\ldots, s_t^r\}$.

 ~~~~~Add the rows $r'_1,\ldots,r'_{t}$ to $M'$, defined as:
	\begin{enumerate}
  	\item[] $M'_{r'_i,j} = \left\{
                                                                        \begin{array}{ll}
                                                                          1, & \hbox{if $M_{r,j} = 1$ and $c(j) = s_i^r$;}  \\
                                                                          0, & \hbox{otherwise.}
                                                                        \end{array}
                                                                      \right.$
  	\end{enumerate}
\item Return $M'$.
\end{enumerate}
\end{framed}

\begin{figure}[!ht]
\centering
\includegraphics[width=0.7\linewidth]{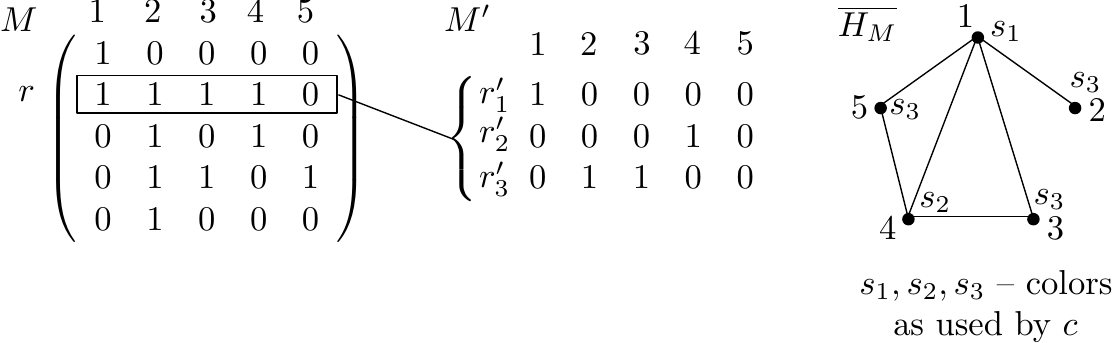}
\caption{An example of a binary matrix $M$, the complement of its undirected containment graph together with an optimal coloring $c$, and a split
of a row of $M$ according to the above algorithm.\label{fig:algorithm-heuristic}}
\end{figure}

\vbox{\begin{theorem}\label{thm:poly-heuristic}
For any $m\times n$ binary matrix $M$, the above algorithm can be implemented to
run in time $O(n^2(n^{1/2}+m))$. The matrix $M'$ output by the algorithm is a conflict-free row split of $M$.
\end{theorem}}

\begin{sloppypar}
\begin{proof}
We first show that the matrix $M'$ produced by the above algorithm is a conflict-free row split of $M$.
Clearly, $M'$ is a row split of $M$.
We say that a row $r'$ of $M'$ is an {\it $r$-row} if $r$ is the row of $M$ such that
$r'$ was added to $M'$ in step 4) of the algorithm when considering row $r$.
Suppose for a contradiction that $M'$ is not conflict-free, and let $\{j,j'\}$ be a pair of conflicting columns of $M'$.
Then, there exist rows $p$, $q$, and $r$ of $M$ and rows
$p'_i$, $q'_{k}$, and $r'_{\ell}$ of $M'$ such that $p_i'$ is a $p$-row,
$q'_{k}$ is a $q$-row, and $r'_{\ell}$ is an $r$-row, $M'_{p'_i,j} = M'_{p'_i,j'} = 1$,
$M'_{q'_k,j} = 1$, $M'_{q'_k,j'} = 0$, and $M'_{r'_\ell,j} = 0$, $M'_{r'_\ell,j'} = 1$.
Since $M'_{p'_i,j} = M'_{p'_i,j'} = M'_{q'_k,j} = M'_{r'_\ell,j'} = 1$,
the definition of $M'$ implies that
$M_{p,j} = M_{p,j'} = M_{q,j} = M_{r,j'} = 1$ and $c(j) = c(j') = s_i^p$,
$c(j) = s_k^q$, and $c(j') = s_\ell^r$.
Consequently, $s_i^p = s_k^q = s_\ell^r$.
Moreover, since  $M'_{q'_k,j'} = 0$ and $c(j') = s_k^q$,
we infer that $M_{q,j'} = 0$ and, similarly,
since $M'_{r'_\ell,j} = 0$ and $c(j) = s^r_\ell$,
we infer that $M_{r,j} = 0$.
It follows that columns $j$ and $j'$ of $M$ are in conflict.
On the other hand, since $c(j) = c(j')$ and $c$ is a proper vertex coloring of $\overline{H_M}$, vertices
corresponding to $j$ and $j'$ are non-adjacent in $\overline{H_M}$. Hence, they are adjacent in $H_M$, which
means that one of the columns $j$ and $j'$ is contained in the other one, contradicting the fact that they are in conflict.
This completes the proof that $M'$ is conflict-free.

It remains to justify the time complexity. The graph $\overline{H_M}$ can be computed in time $O(mn^2)$ by comparing every ordered pair of columns
for containment.  In the same time, we can also compute a transitive orientation $\sqsubset$ of $H_M$ (as done for example in Observation~\ref{obs:transitive}).
An optimal coloring of $\overline{H_M}$ corresponds to a minimum  chain partition of the partially ordered set $P = (V(H_M),\sqsubset)$,
and can be computed in time $O(n^{5/2})$ as follows (cf.~\cite{MR545530} and~\cite[p.~73-74]{makinen2015genome}).
Applying the approach of Fulkerson~\cite{MR0078334}, a minimum chain partition of $P$ can be computed by solving a maximum matching problem in
a derived bipartite graph having $2n$ vertices. This can be done in time $O(n^{5/2})$ using the algorithm of Hopcroft and Karp~\cite{MR0337699}.
Step 4) of the algorithm can be executed in time $O(\sum_{r}|c(V(G_{M,r})|n) = O(mn^2)$.
The claimed time complexity of $O(n^2(n^{1/2}+m))$ follows.
\end{proof}
\end{sloppypar}

\section{Implementation and experimental results}\label{sec:experiments}

C++ implementations of both algorithms are available at \url{https://github.com/alexandrutomescu/MixedPerfectPhylogeny}. The input matrices must be in .csv format, and are allowed to have duplicate columns. In addition to a conflict-free row split matrix, we also output its perfect phylogeny tree. We tested our implementation on the ten binary matrices constructed from clear cell renal cell carcinomas (ccRCC) from~\cite{Gerlinger:2014aa}: EV001-EV003, EV005-EV007, RMH002, RMH004, RMH008, RK26. The phylogenetic trees constructed by~\cite{Gerlinger:2014aa} from these matrices appear in~\cite[Fig.3]{Gerlinger:2014aa}.

\begin{table*}[h!]
\caption{The numbers of rows, columns, and pairwise distinct columns in the input matrices, the lower bounds on the minimum number of rows in conflict-free row splits of the input matrices given by Corollary~\ref{cor:WABI}, and the numbers of rows and pairwise distinct rows in the matrices output by the heuristic algorithm. The distinct rows form the leaves of the output perfect phylogeny.\label{tab:number-of-rows}}
\centering
\bigskip
\begin{tabular}{l|ccc|c|cc}\hline
\multicolumn{1}{c|}{} & \multicolumn{3}{c|}{\bf Input} & \multicolumn{1}{c|}{lower bound} & \multicolumn{2}{c}{\bf Output} \\
Name & \#rows & \#cols & \#distinct cols  &  on \#rows in output & \#rows & \#distinct rows\\\hline
EV001 & 10 & 122 & 22 & 37 & 51 & 22 \\
EV002 & 7 & 103 & 17 & 18 & 29 & 17 \\
EV003 & 8 & 56 & 12 & 13 & 19 & 12 \\
EV005$^{\ast}$ & 7 & 83 & 10 & 7 & 7 & 7 \\
EV006 & 9 & 76 & 11 & 13 & 25 & 11 \\
EV007 & 8 & 66 & 15 & 19 & 25 & 15 \\
RHM002 & 5 & 54 & 11 & 9 & 13 & 11 \\
RHM004 & 6 & 140 & 17 & 16 & 21 & 17 \\
RHM008 & 8 & 81 & 10 & 10 & 20 & 10 \\
RK26 & 11 & 75 & 17 & 18 & 26 & 17 \\
\hline
\end{tabular}\\
\smallskip
{\scriptsize $^{\ast}$Input is conflict-free}
\end{table*}

These ten matrices have between 5 and 11 rows, and between 55 and 140 columns. One of them is already conflict-free, while the other nine do not belong to the polynomially-solvable case discussed in Section~\ref{sec:no-containments-in-conflicting}. On each of these nine matrices, our heuristic algorithm from Section~\ref{sec:heuristic} runs in less than one second. On ten random binary matrices with 50 rows and 1000 columns, the heuristic algorithm runs on average in 97 seconds. Due to the restricted structure of the input matrices on which the polynomial time algorithm from Section~\ref{sec:no-containments-in-conflicting} works correctly, we did not test the implementation of this algorithm on random inputs. However, since this algorithm is simpler than the heuristic one, it is plausible to expect that it will run at least as fast. Of course one might want to first check whether the input $m\times n$ binary matrix satisfies the assumption that no column is contained in both columns of a pair of conflicting columns. This can be tested straightforwardly in time $O(n^2(m+n))$ by first classifying each ordered pair of columns as conflicting, disjoint, or in containment, and then testing each triple of columns for the condition that the first one is contained in each of the other two, which are in conflict. The running time of this check is asymptotically worse than that of the heuristic algorithm. However, we expect the running times to differ little on practical instances of moderate sizes, since the constant hidden in the $O()$ notation for the above check is small.

In Table~\ref{tab:number-of-rows} we list the numbers of rows (i.e., samples) in the ten original matrices, together with numbers of columns and pairwise distinct columns, the lower bounds on the minimum number of rows in a conflict-free row split of each of the matrices given by Corollary~\ref{cor:WABI}, and the numbers of rows and pairwise distinct rows in the matrices output by the heuristic algorithm.
The similarity between the numbers of pairwise distinct columns in the input and of pairwise distinct rows in the output can be explained by the observation that if the input matrix consists of $n$ pairwise distinct columns, then there will be at most $n$ distinct rows
in the naive solution (split each row into as many rows of the identity matrix as the number of $1$s it contains).
Thus $n$ is an upper bound for an optimal solution to the {\sc Minimum Distinct Conflict-Free Row Split} problem, which explains why our heuristic algorithm applied to the {\sc Minimum Distinct Conflict-Free Row Split} problem performs similarly as the naive one. In Fig.~\ref{fig:trees} we show four perfect phylogeny trees corresponding to the matrices output by the heuristic algorithm. The results on all matrices are available online, linked from \url{https://github.com/alexandrutomescu/MixedPerfectPhylogeny}.

We also ran our heuristic algorithm on the same datasets, with the difference that we removed from the input matrices those columns that appear (as binary vectors) strictly less than 2 times (this value is a parameter to our implementation). This is the same idea and default threshold used by Popic et al.~\cite{Popic2015} in their tool LICHeE. Their motivation is that several mutations accumulate before a new tumor branch separates, and thus such rare mutation patterns may be due to errors in the data. We show four output trees in Fig.~\ref{fig:trees-minsupport2}.

\begin{sloppypar}
Finally, we also ran LICHeE on the same ten patients from~\cite{Gerlinger:2014aa}. Since LICHeE uses the variant allele frequency (VAF) of every mutation, we used the matrices containing VAF values linked from
\url{https://github.com/viq854/lichee/tree/master/LICHeE/data/ccRCC}
and ran LICHeE with the parameters indicated therein. As referred to in the above, LICHeE starts by grouping the robust mutations into clusters of size at least a given number, by default 2. (This is the same experiment as the one done in the paper \cite{Popic2015} introducing LICHeE; see~\cite{Popic2015} for further details.) We show four trees produced by LICHeE in Fig.~\ref{fig:trees-lichee}. Note LICHeE does not necessarily output binary trees.
\end{sloppypar}

\begin{figure*}
\centerline{
\subfigure[EV003]{
\includegraphics[width=4cm]{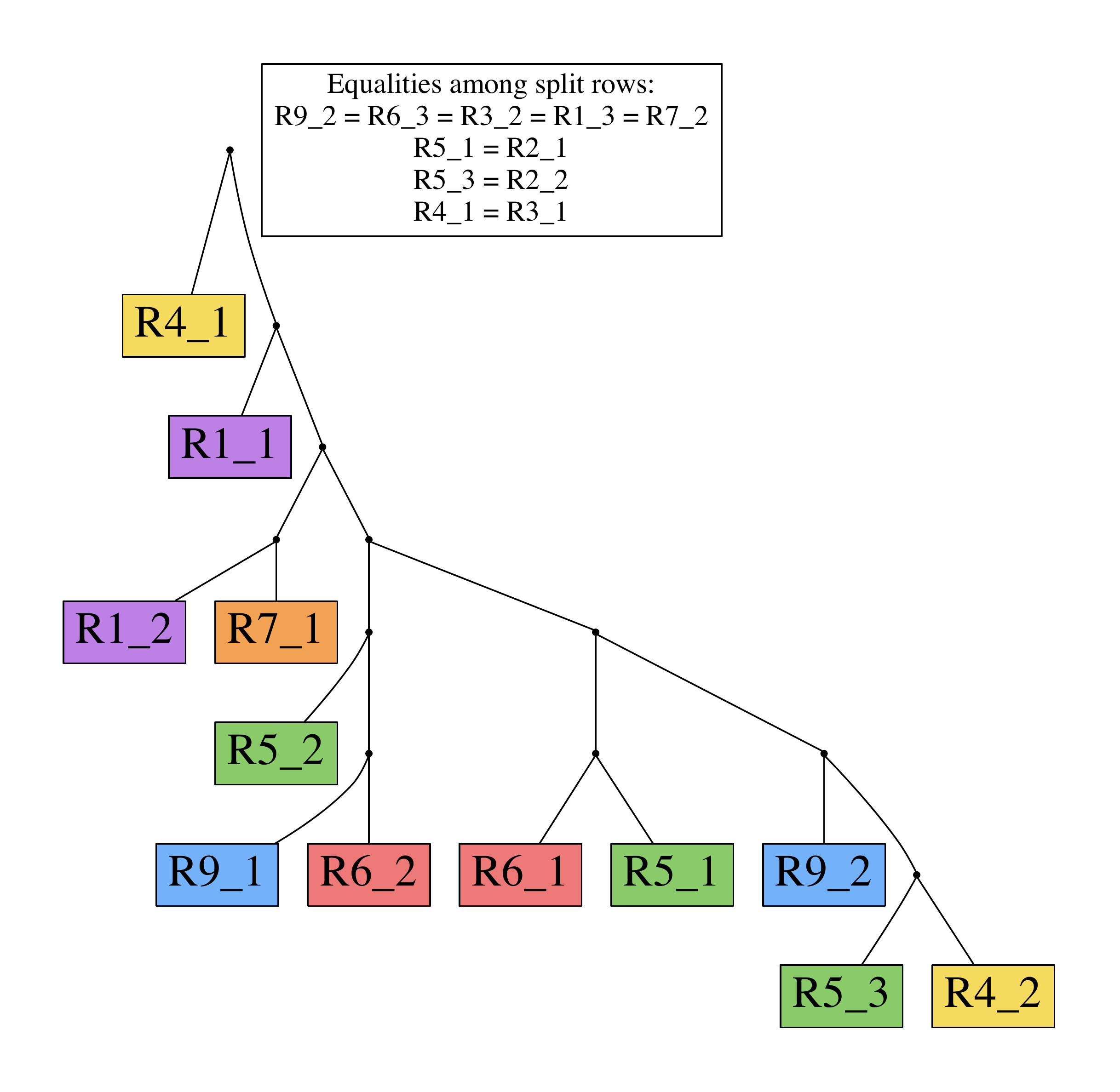}
}
\hfill
\subfigure[RMH002]{
\includegraphics[width=3.5cm]{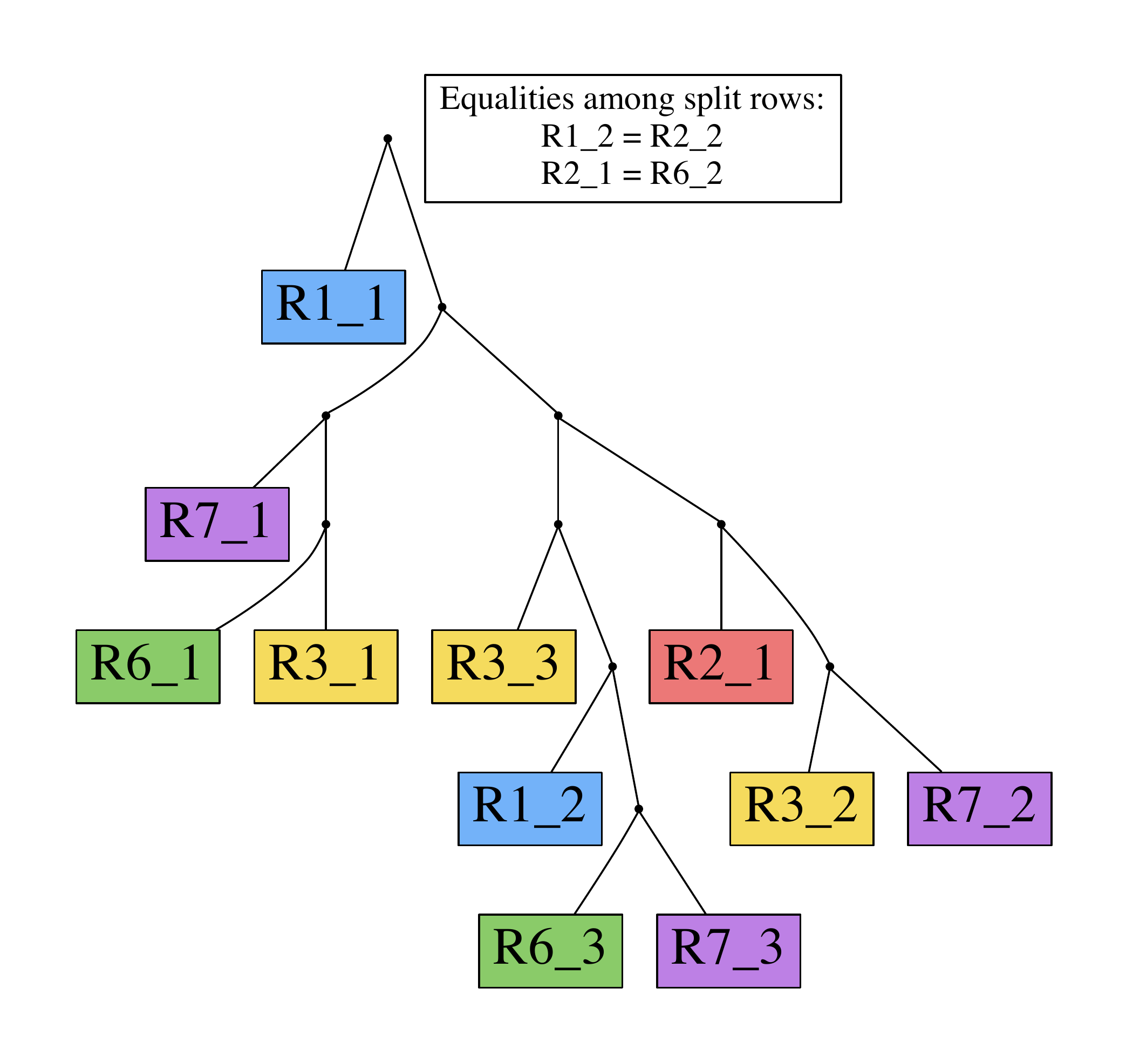}
}
\hfill
\subfigure[RMH008]{
\includegraphics[width=4cm]{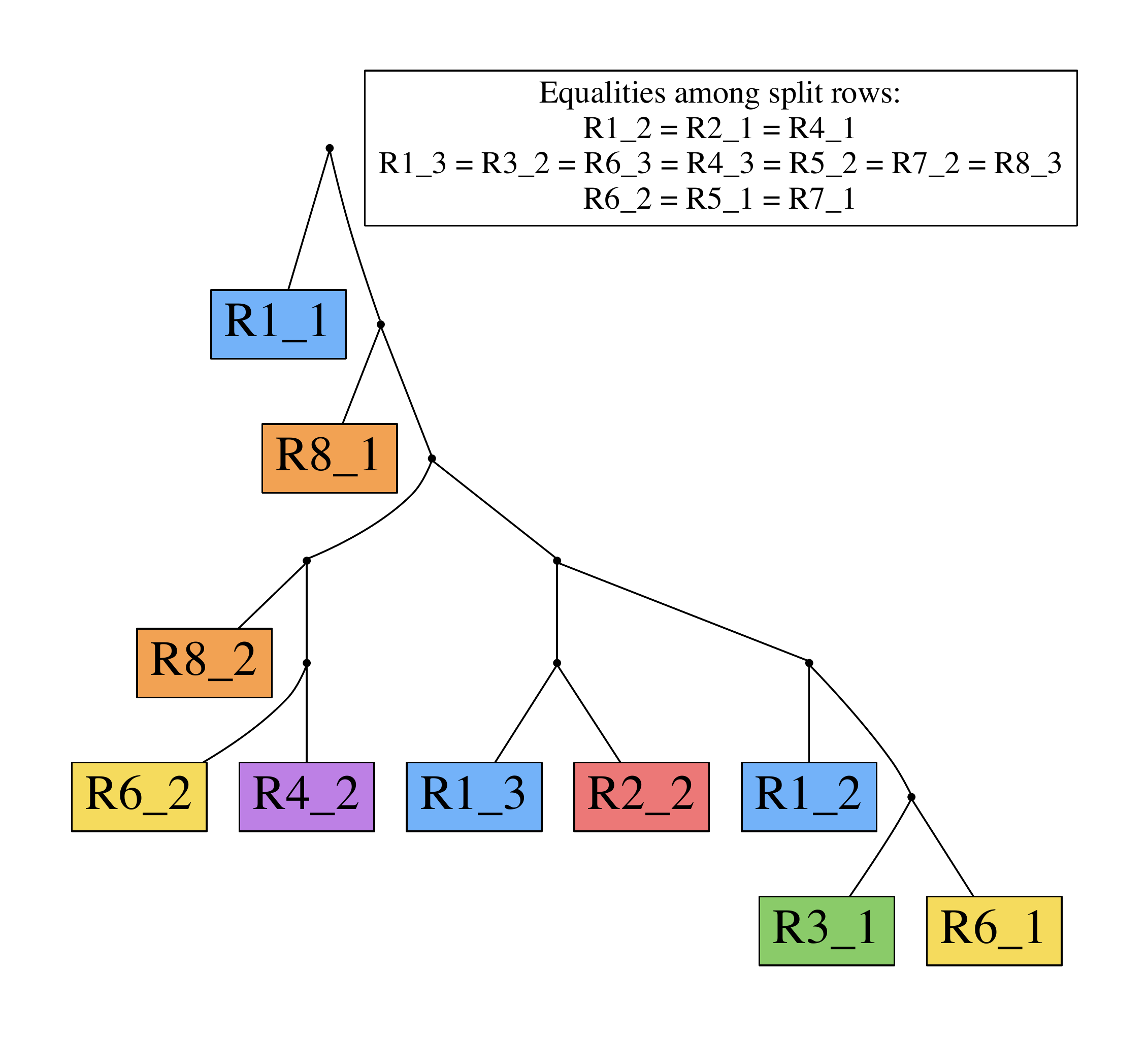}
}
\hfill
\subfigure[RK26]{
\includegraphics[width=5.5cm]{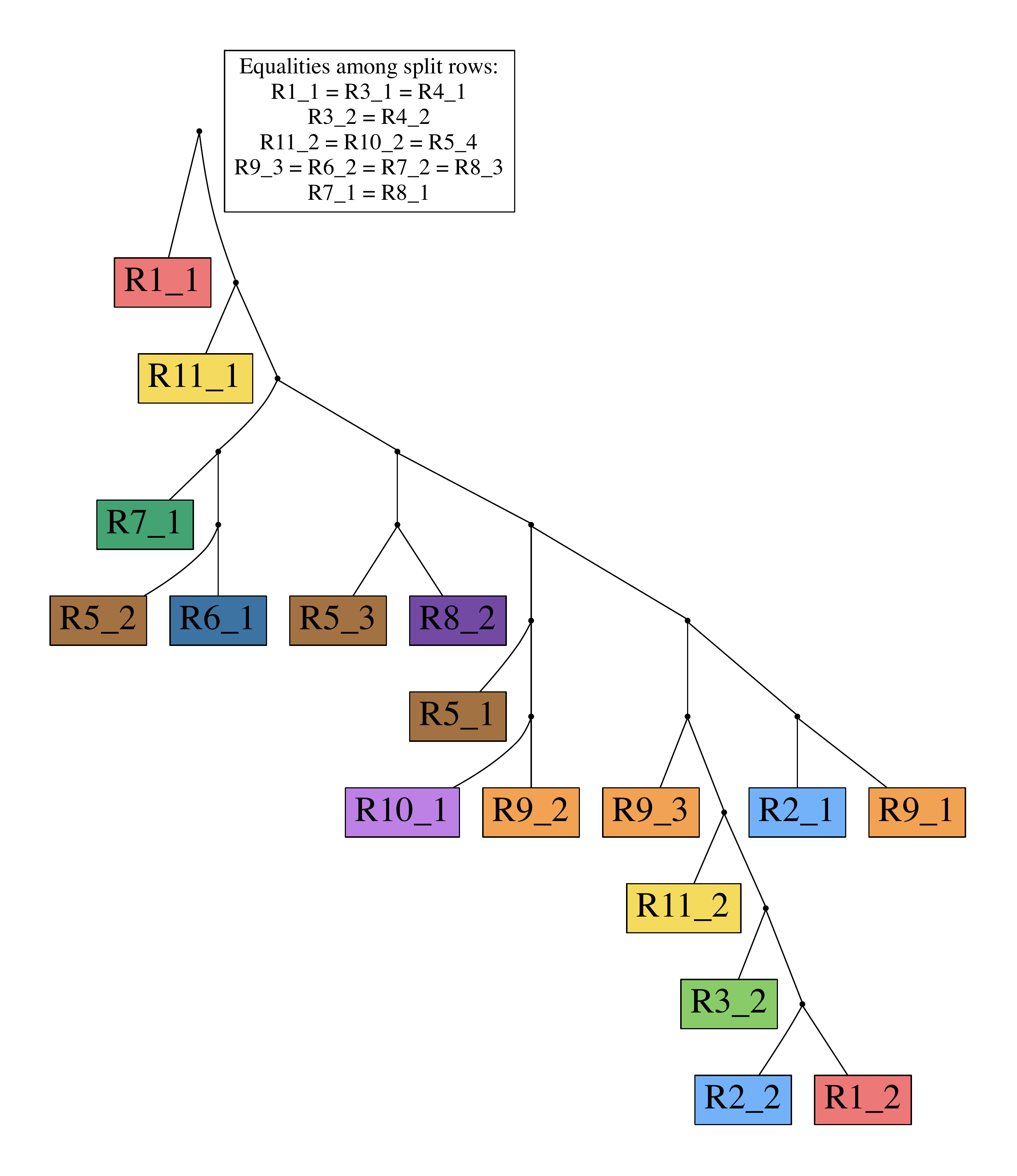}
}
}
\caption{Four perfect phylogeny trees corresponding to the conflict-free row split matrices output by our heuristic algorithm from Section~\ref{sec:heuristic}. The naming convention is: R1 is an original row (i.e., sample) name, and R1\_1, R1\_2, R1\_3 are the row (i.e., samples) names corresponding to R1 in the output matrix. Equal split rows form a single node of the perfect phylogeny (equalities are indicated in boxes).\label{fig:trees}}
\end{figure*}

\begin{figure*}
\centerline{
\subfigure[EV003]{
\includegraphics[width=3cm]{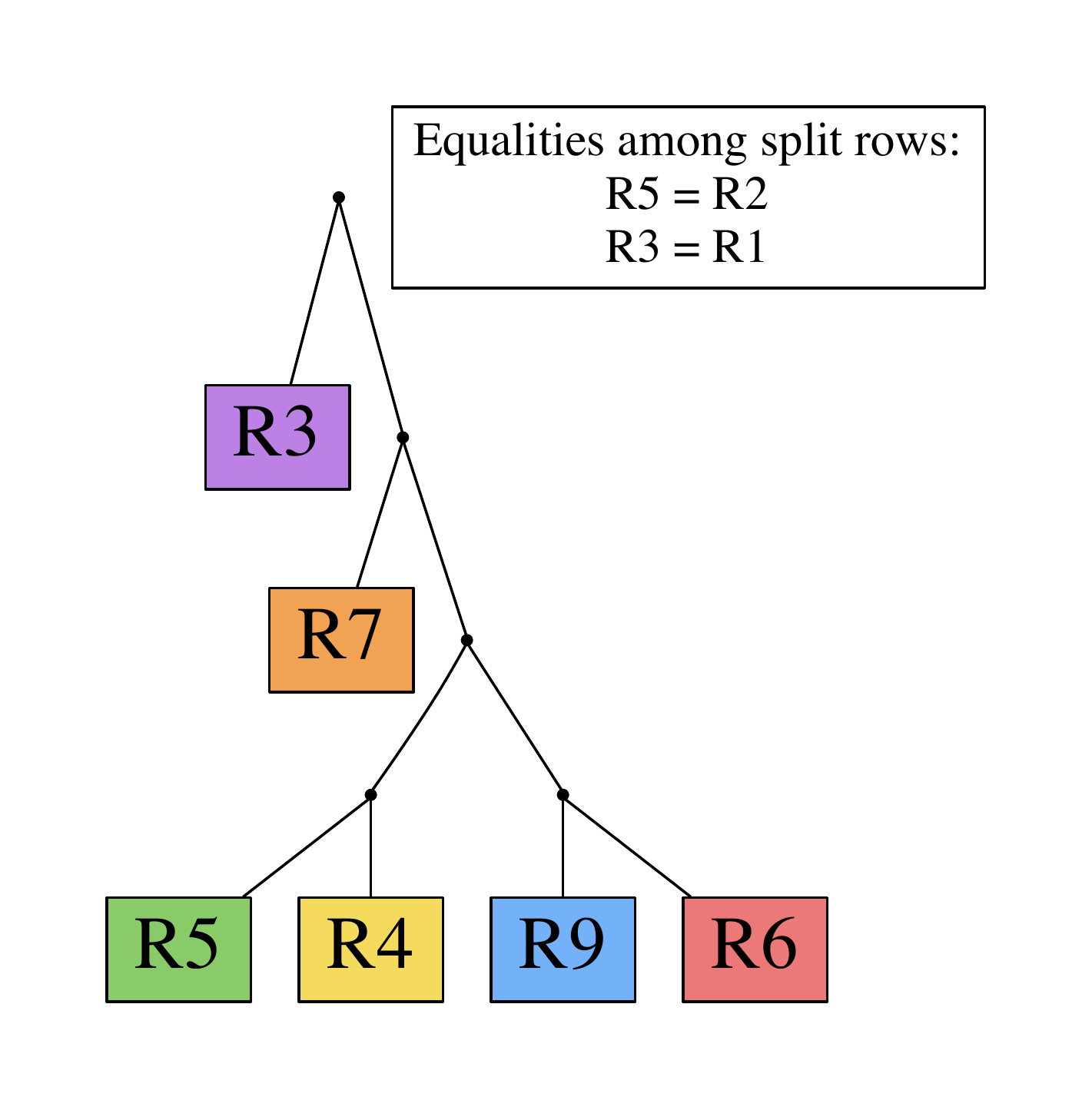}
}
\hfill
\subfigure[RMH002]{
\includegraphics[width=2cm]{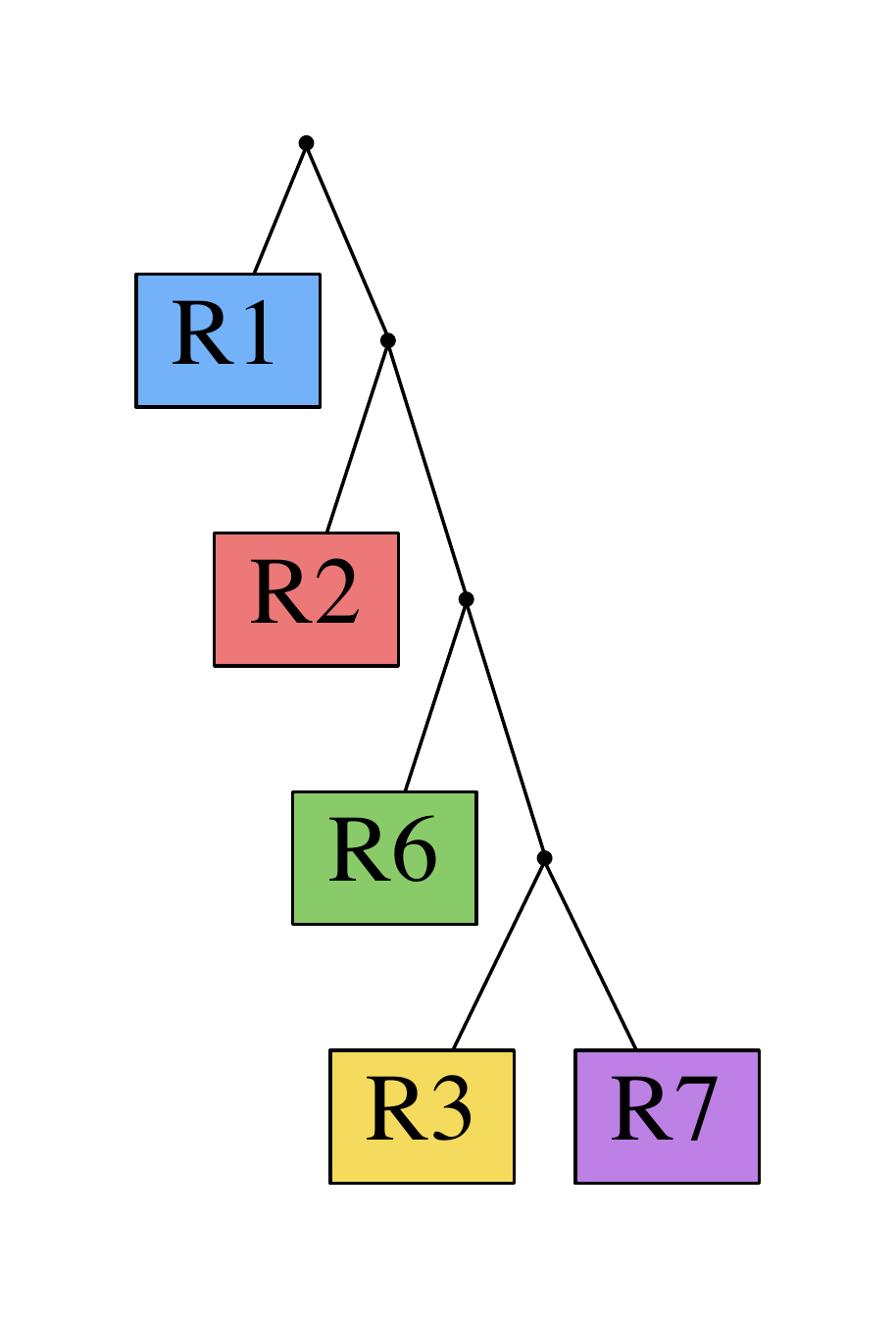}
}
\hfill
\subfigure[RMH008]{
\includegraphics[width=5.5cm]{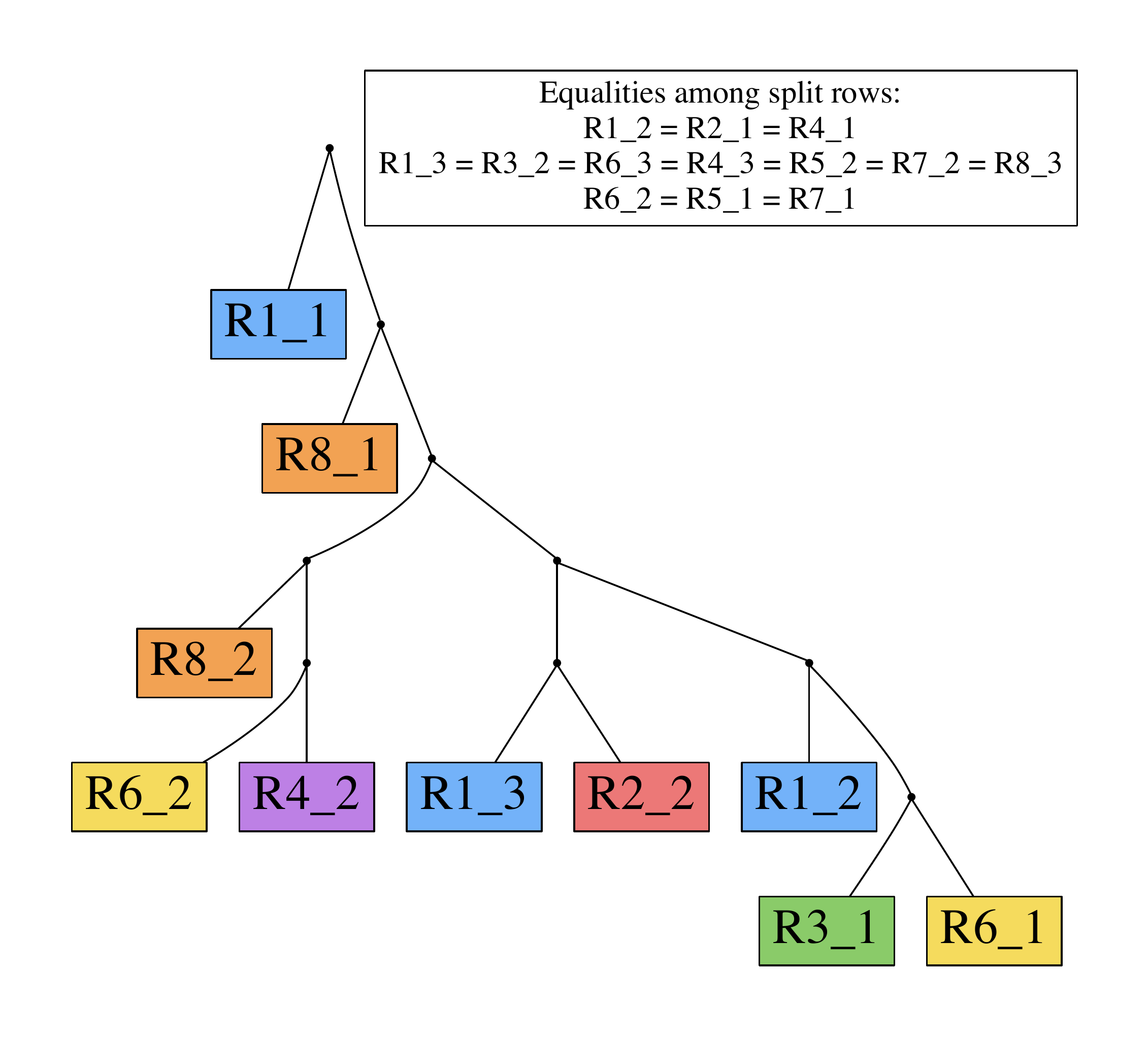}
}
\hfill
\subfigure[RK26]{
\includegraphics[width=5.5cm]{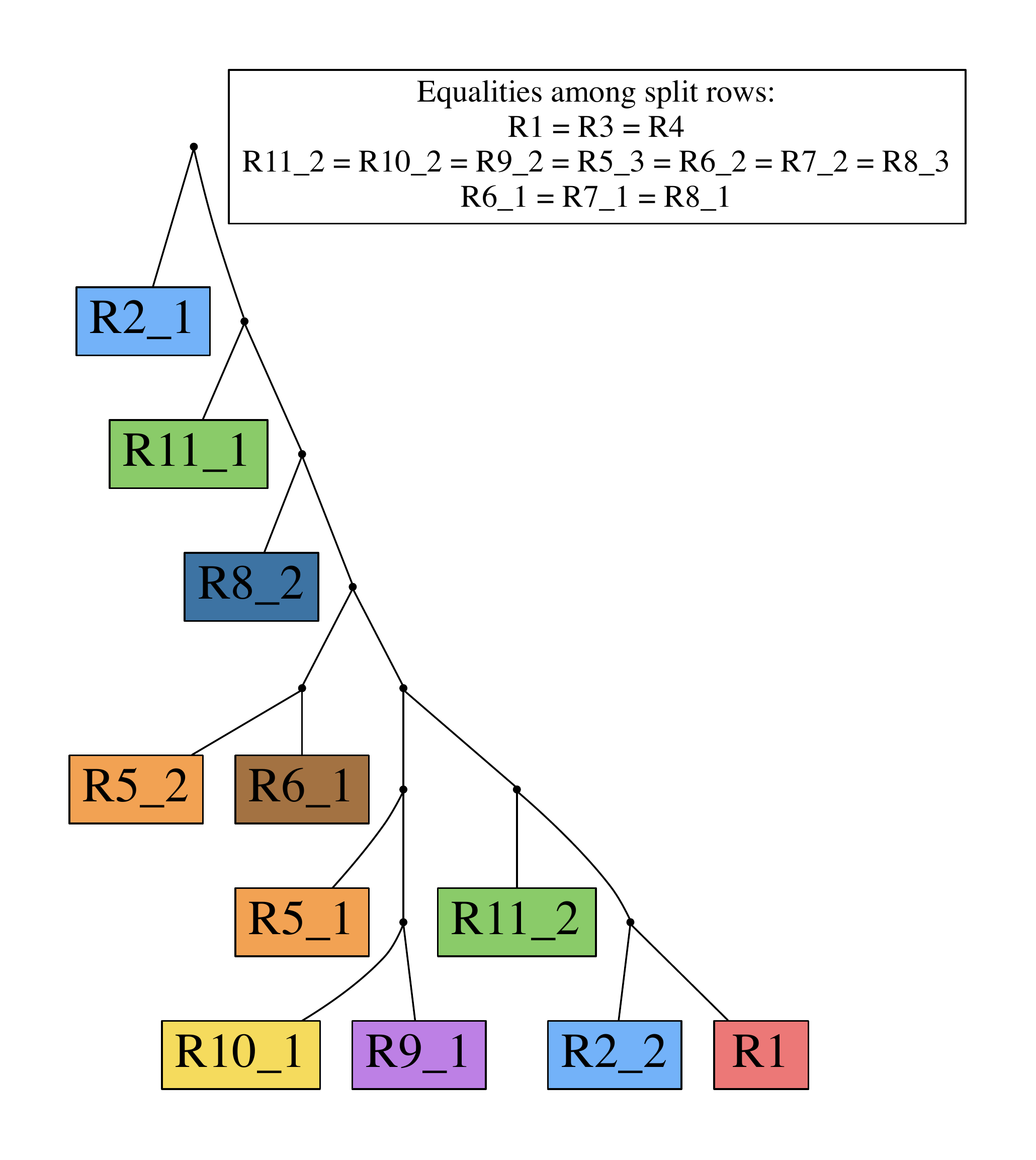}
}}
\caption{The perfect phylogeny trees output by our heuristic algorithm, after removing from the input matrices those columns that appear only once. The naming convention is as in Fig.~\ref{fig:trees}. Compare these trees to the trees from \cite[Fig.~3]{Gerlinger:2014aa} and to the ones from Fig.~\ref{fig:trees-lichee}.\label{fig:trees-minsupport2}}
\end{figure*}

\begin{figure*}
\centerline{
\subfigure[EV003]{
\includegraphics[width=4.5cm]{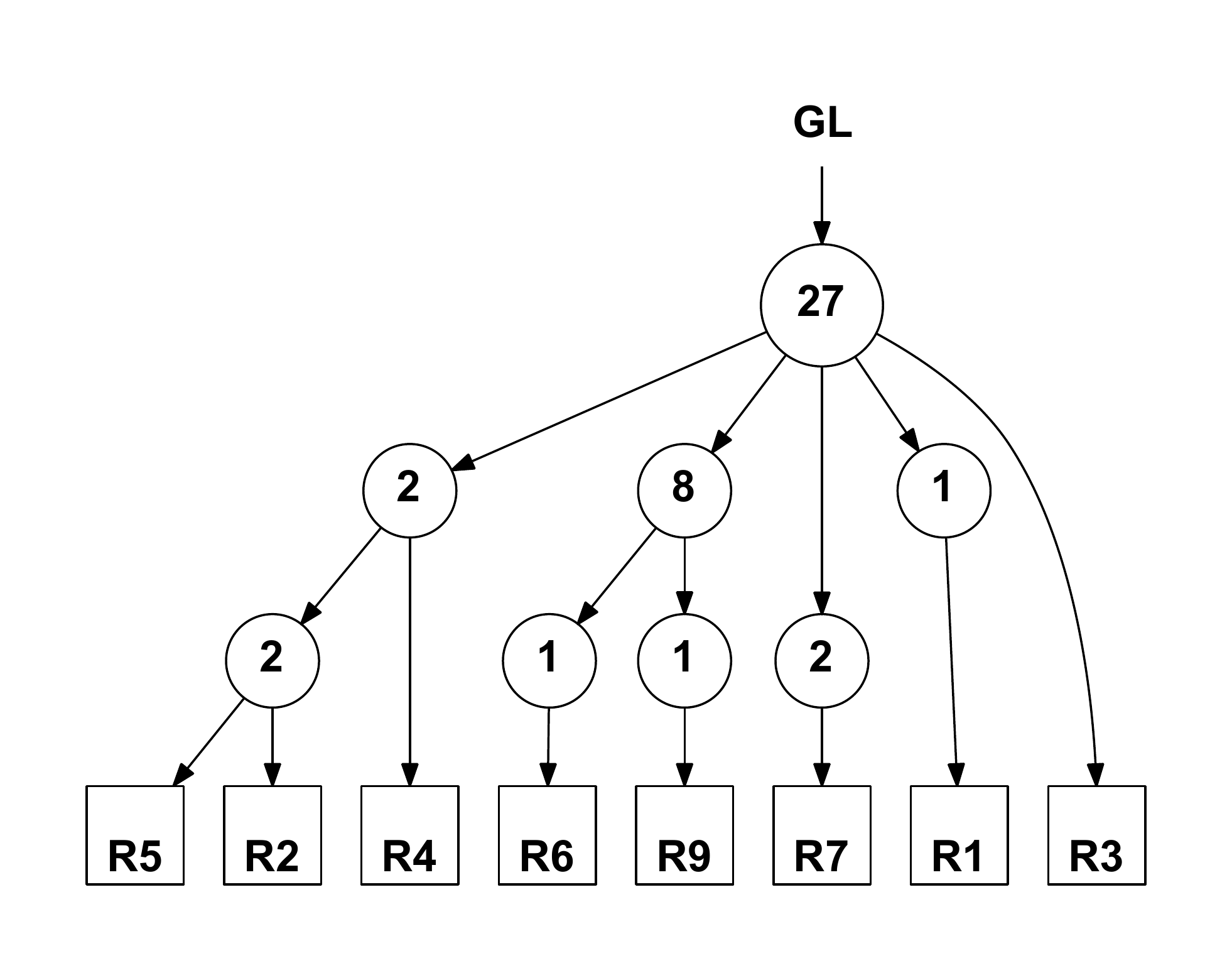}
}
\hfill
\subfigure[RMH002]{
\includegraphics[width=3cm]{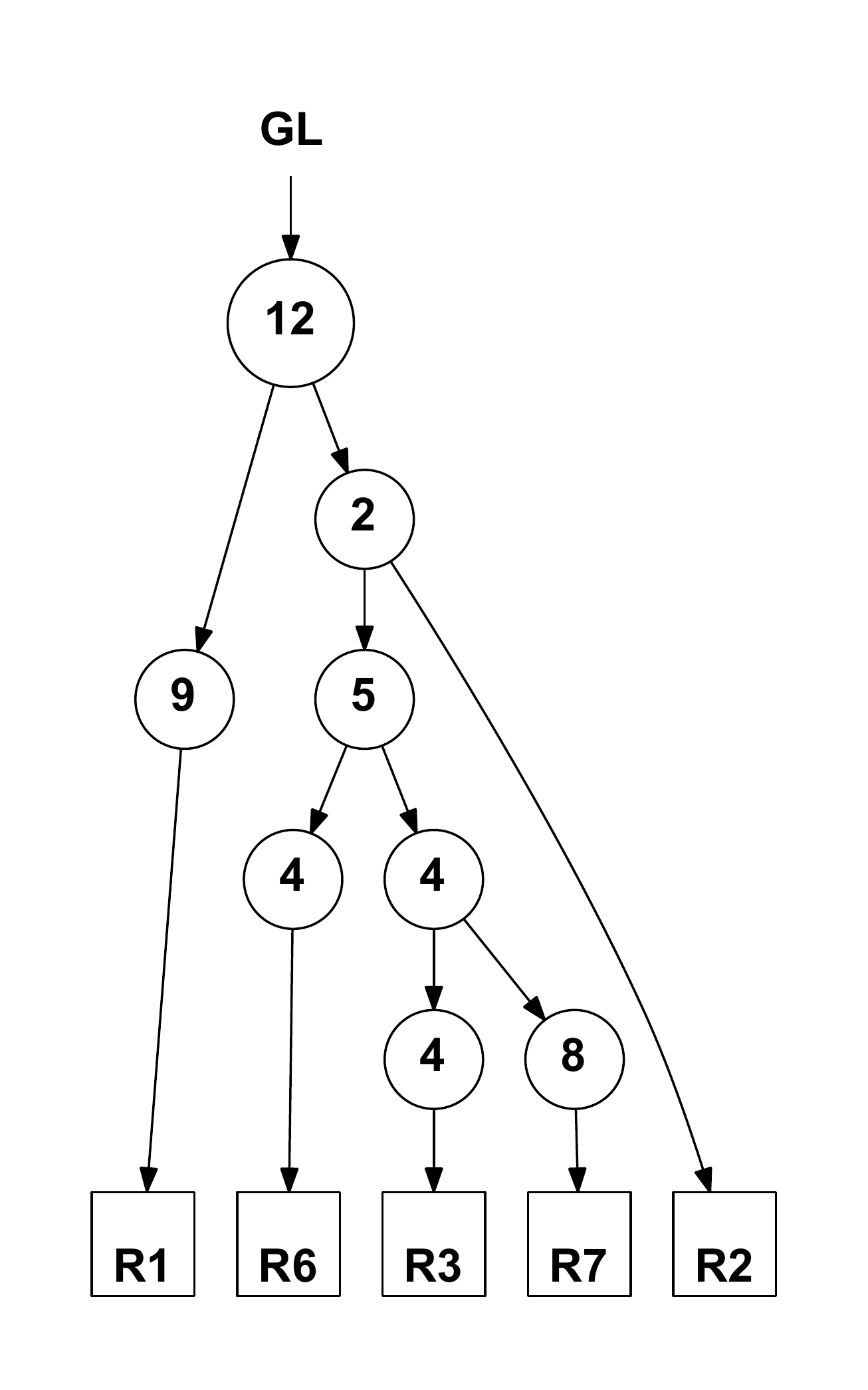}
}
\hfill
\subfigure[RMH008]{
\includegraphics[width=4cm]{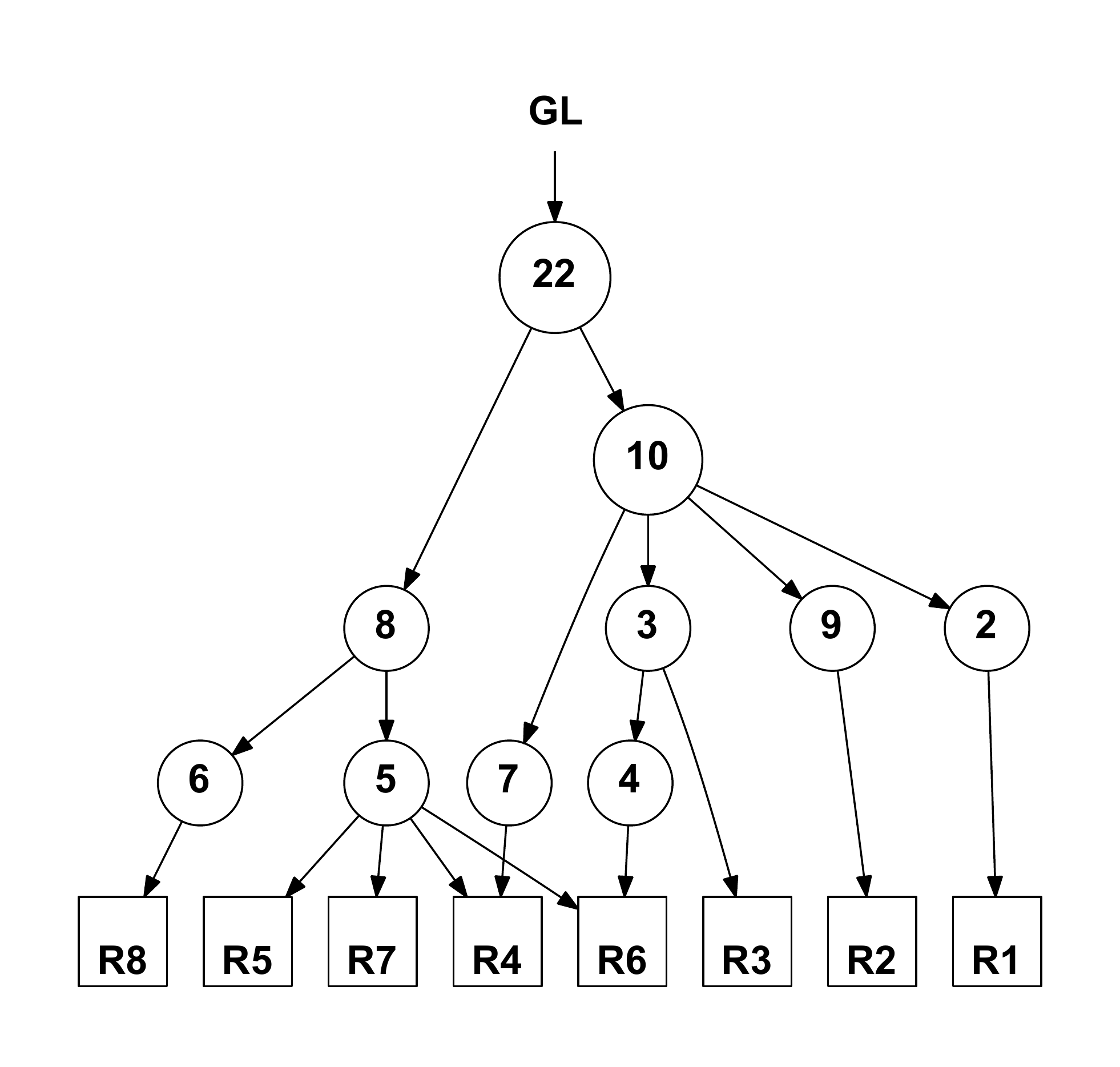}
}
\hfill
\subfigure[RK26]{
\includegraphics[width=5cm]{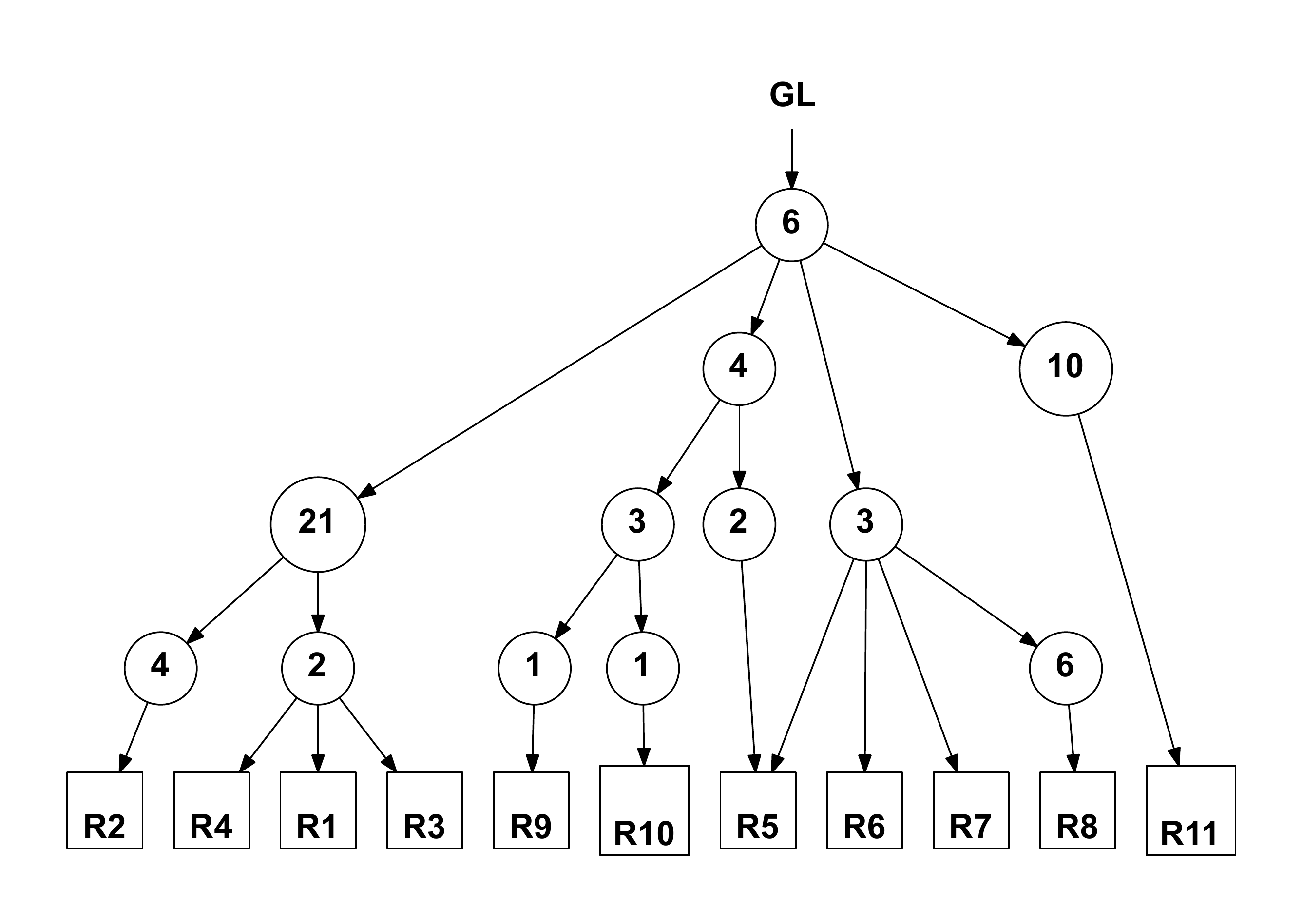}
}
}
\caption{The tumor evolutions predicted by LICHeE~\cite{Popic2015}. Numbers in the circles represent the number of mutations at that node. See~\cite{Popic2015} for further details.\label{fig:trees-lichee}}
\end{figure*}

\section{Discussion}
\label{sec:discussion}

After filtering the unique columns of the matrices EV003 and RMH002, the resulting matrices are conflict-free. The perfect phylogeny for RMH002 (Fig.~\ref{fig:trees-minsupport2}) is the same as the one produced by LICHeE (Fig.~\ref{fig:trees-lichee}) and the one from~\cite[Fig.~3]{Gerlinger:2014aa}. The perfect phylogeny for EV003 (Fig.~\ref{fig:trees-minsupport2}) is almost identical to the one produced by LICHeE (Fig.~\ref{fig:trees-lichee}) and the one from~\cite[Fig.~3]{Gerlinger:2014aa}, the only difference being that each of the pairs of samples $\{$R2, R5$\}$ and $\{$R1, R3$\}$ is collapsed into a single leaf of the phylogeny. This may be due to the fact that LICHeE and the method used by~\cite{Gerlinger:2014aa} exploit VAF values, not only binary values. Overall, these two matrices suggest that filtering out rare columns may be a relevant strategy.

All the columns of the matrix RMH008 appear at least two times, thus there are no columns to be filtered out. While LICHeE finds that only samples R4 and R6 are a combination of more leaves of the tumor phylogeny, our algorithm finds that all eight samples are combinations of two or more leaves. However, there are similarities to the prediction of LICHeE. For example, both samples R4 and R6 have mutations in common with R5 and R7 (R6\_2 = R5\_1 = R7\_1, (R1\_3 = ) R4\_3 = R5\_2 = R7\_2). Our subclones R6\_2 and R4\_3 also correspond to the subclones R6\_{dom} and R4\_{min}, respectively, found in~\cite{Gerlinger:2014aa}. Samples R4 and R6 also have mutations in common with R1 (R1\_3 = R6\_3 = R4\_3), and with R2 (nodes R2\_2 and R1\_3 = R6\_3 = R4\_3 are both siblings of a node at distance 4 from the root). The subclones R6\_3 and R4\_3 also correspond to the subclones R6\_{min} and R4\_{dom}, respectively, found in~\cite{Gerlinger:2014aa}.

In matrix RK26, 6 columns are unique, and they have been filtered out in Fig.~\ref{fig:trees-minsupport2}. LICHeE finds that only sample R5 is a combination of more leaves, while our algorithm again finds that all samples are combinations of two or more leaves. However, there are again similarities in the results. Sample R5 has mutations shared with R6, R7, R8 (nodes R5\_2 and R6\_1 = R7\_1 = R8\_1 are siblings of a node at depth 4). Subclone R5\_2 also corresponds to subclone R5\_{dom} found in~\cite{Gerlinger:2014aa}. Sample R5 also has mutations shared with R9 and R10 (nodes R5\_1, R10\_1 and R9\_1 have the lowest common ancestor at depth 5). Subclone R5\_1 also corresponds to subclone R5\_{min} found in~\cite{Gerlinger:2014aa}.

\section{Conclusion}\label{sec:conclusion}

In this paper we {showed hardness of the
{\sc Minimum Conflict-Free Row Split} and the
{\sc Minimum Distinct Conflict-Free Row Split} problems, }and gave a polynomial time algorithm for the
{\sc Minimum Conflict-Free Row Split} problem on instances such that no column is contained in both columns of a pair of conflicting columns.
More general tractable instances could be found by inspecting further dependencies between column containment and conflictness. For example, it remains open whether the {\sc Minimum Conflict-Free Row Split} problem is tractable on matrices in which no pair of conflicting columns is contained in both columns of a pair of conflicting columns. It would also be interesting to
identify polynomially solvable cases of the {\sc Minimum Distinct Conflict-Free Row Split} problem and to
explore variations of the problems in which we are also allowed to edit the entries of the input matrix.

{In the paper we also gave a polynomial time heuristic algorithm for the {\sc Minimum Conflict-Free Row Split} problem based on graph coloring.
We leave as a question for future research to determine the \hbox{(in-)approximability} status of the optimization variants of the two problems. }
Finally, we remark that in~\cite{DBLP:conf/wabi/HajirasoulihaR14} it was assumed that the matrices have no duplicated columns, which was not necessary in this paper.

%
%

\bibliographystyle{abbrv}
\bibliography{biblio}
\end{document}